\documentclass[11pt,twoside]{article}
\usepackage{fancyhdr}
\usepackage{amsmath, amsthm, amssymb}
\usepackage{graphicx}
\usepackage{subcaption}
\usepackage{float}
\usepackage{booktabs}
\usepackage{multirow}
\usepackage{longtable}
\usepackage{dblfloatfix}
\usepackage{rotating}
\usepackage{authblk}

\usepackage{tikz}
\usetikzlibrary{positioning, arrows.meta, calc, shapes.geometric}
\usepackage[font=sf, labelfont={sf,bf}]{caption}
\usetikzlibrary{backgrounds}

\theoremstyle{plain}
\newtheorem{lemma}{Lemma}

\newtheorem{theorem}{Theorem}
\newtheorem{corollary}{Corollary}

\theoremstyle{definition}

\theoremstyle{remark}
\newtheorem{remark}{Remark}

\usepackage[utf8]{inputenc}

\usepackage{xcolor}
\usepackage{todonotes}

\usepackage[margin=1in]{geometry}

\usepackage[colorlinks=true, linkcolor=blue, citecolor=blue]{hyperref}
\usepackage{cleveref}

\usepackage[numbers, sort&compress]{natbib}

\DeclareMathOperator*{\argmin}{arg\,min}

\newcommand{\vect}[1]{\boldsymbol{#1}}
\newcommand{\R}{\mathbb{R}}

\newcommand{\headers}[2]{%
  \fancyhead[CO]{#1}
  \fancyhead[CE]{#2}
}

\headers{Discovery and control of multistable systems via structured Neural ODEs}%
        {I. Griss Salas, E. King}

\title{Data-driven discovery and control of multistable nonlinear systems and hysteresis via structured Neural ODEs}

\author[1]{Ike Griss Salas\thanks{Corresponding author: grisal@uw.edu}}

\author[2]{Ethan King}

\affil[1]{Department of Applied Mathematics, University of Washington, Seattle, WA, USA}
\affil[2]{Pacific Northwest National Laboratory, Richland, WA, USA}

\date{}

\begin{document}
\maketitle

\begin{abstract}
Many engineered physical processes exhibit nonlinear but asymptotically stable dynamics that converge to a finite set of equilibria determined by control inputs. Identifying such systems from data is challenging: stable dynamics provide limited excitation and model discovery is often non-unique. We propose a minimally structured Neural Ordinary Differential Equation (NODE) architecture that enforces trajectory stability and provides a tractable parameterization for multistable systems, by learning a vector field in the form \(F(x,u) = f(x)\,(x - g(x,u))\), where \(f(x) < 0\) elementwise ensures contraction and \(g(x,u)\) determines the multi-attractor locations. Across several nonlinear benchmarks, the proposed structure is efficient on short time horizon training, captures multiple basins of attraction, and enables efficient gradient-based feedback control through the implicit equilibrium map \(g\). 

\end{abstract}

\section{Introduction}
System identification and control of dynamical systems from observables remain ubiquitous challenges across physics and engineering. Data-driven modeling as an approach has gained popularity but fundamental questions regarding the limits of learnability persist. Recent discoveries, notably Shumaylov et al. \cite{shumaylov2025discoverabledatadiscoveryrequires}, have established that \emph{chaos} is a necessary condition for the \emph{unique} discovery of dynamical systems from a single observed trajectory. This poses a significant constraint on the class of engineering and biological systems that operate in non-chaotic regimes, implying the ``true'' underlying model may be formally unidentifiable from trajectory data alone.

The lack of uniqueness motivates a critical pivot in perspective: if the unique recovery of governing equations is ill-posed for non-chaotic systems, can we instead leverage this non-uniqueness to learn representations that are \emph{useful}?

The concept of learning favorable representations is not new; indeed, it is fundamental to classical mechanics and astronomy. The evolution from Cartesian coordinates to angular dynamics for the pendulum, or the development of the heliocentric model, demonstrates that the choice of state representation dictates the tractability and interpretability of the problem \cite{Goldstein2002}. In both cases, the ``useful'' representation preceded the full derivation of the governing laws, suggesting that finding the right geometry is often a prerequisite for understanding the dynamics.

However, in the context of data-driven approaches, there remains a significant gap in discovering representations explicitly constructed for downstream tasks such as interpretability, stability, and control. 

In this work, we investigate the discovery of useful representations within the specific domain of \emph{asymptotically stable dynamical systems}. These systems, characterized by bounded trajectories that converge to unique steady states determined by control input, are pervasive in engineering systems, e.g. chemical, biological, and manufacturing processes. We introduce a data-driven methodology that infers these dynamics solely from observable trajectories, demonstrating success with sparse temporal data (time horizons significantly shorter than full system evolution). Specifically, our framework enables:
\begin{enumerate}
    \item \textbf{Trajectory Stability:} The enforcement of bounded trajectory learning consistent with the observation of asymptotic stability. 
    \item \textbf{Interpretability:} The identification of salient dynamical features, such as bifurcation points, e.g. ``tipping points'', and decoupling of multivariate dynamics into univariate representations when appropriate. 
    \item \textbf{Multistability:} The ability to capture multiple steady states, the number of which need not be known \textit{a priori.}
    \item \textbf{Control:} The ability to construct feedback control policies capable of traversing non-trivial bifurcations, such as hysteresis loops, while easily adhering to user provided control constraints. 
\end{enumerate}

We demonstrate the effectiveness of this method on multiple nonlinear systems.

\section{Background and related work}
Model discovery, or more specifically the discovery of continuous-time dynamical systems from time-series observations, has been a long-standing problem \cite{aastrom1971system, ljung2010perspectives, keesman2011system}, with seminal works using symbolic regression \cite{bongard2007automated,schmidt2009distilling} to obtain analytical solutions. The works in \cite{wang2011,Brunton2016,rudy2017data,schaeffer2017learning} further popularized the use of sparsity priors, extending ideas from compressive sensing literature \cite{donoho2006compressed,candes2006robust,candes2008intro}. 

The use of deep learning architectures has also gained popularity for system identification \cite{PILLONETTO2025111907, LJUNG20201175}, with implicit neural models gaining particular traction. The work in \cite{chen2018neural} introduced the Neural Ordinary Differential Equation (NODE) framework, a learning paradigm for governing dynamical systems by integrating neural network–parameterized dynamics and comparing simulated trajectories against time-series observations. The associated loss is minimized by tuning the models by backpropagating through a selected numerical integrator. NODEs have since been used for latent variable modeling and continuous normalizing flows and have shown considerable success in a variety of engineering tasks, e.g., \cite{Koch2025Neural,koch2025contactdatadrivenfrictionstirprocess, koch2024atmospheric,koch2024gnde}.

Despite their versatility, NODEs can be difficult to train. The associated loss landscape arising from training is often extremely non-convex. For example, \cite{okamoto2025learningsimplestneuralode} illustrated that even training on a simple one-dimensional linear ODE can suffer from severe instabilities when optimized using gradient descent. Continuous-time dynamics make the loss landscape highly ill-conditioned, necessitating careful architecture design, regularization, or additional constraints on the learning objective.

To mitigate these challenges, a popular extension of NODEs is to incorporate Lyapunov stability. The work in \cite{Kolter2019LearningSD} jointly learns both the dynamics and a Lyapunov candidate function, guaranteeing provable stability over the entire state space. Extensions such as \cite{ wu2023neural, sochopoulos2024learningdeepdynamicalsystems,kang2021stableneuralodelyapunovstable} impose Lyapunov conditions directly in the training objective.

Additional efforts focus on designing intrinsically stable NODE architectures. The work \cite{sochopoulos2024learningdeepdynamicalsystems} proposes a class of stable latent dynamical systems compatible with NODE-style training, imposing a Lyapunov-like structure to obtain provably stable dynamics with multiple attractors. The work in \cite{kang2021stableneuralodelyapunovstable} addresses stability from the perspective of adversarial robustness, building on the observation that NODEs are more resilient to perturbations than standard deep networks \cite{yan2020robustness} and accordingly impose a Lyapunov-stable learning criterion. Other works incorporate dissipativity constraints by enforcing a dissipative structure in learned models, projecting arbitrary neural-network dynamics onto a dissipative class using a generalized Kalman–Yakubovic–Popov (KYP) lemma \cite{okamoto2024learningdeepdissipativedynamics}.

An emerging alternative to continuous ODE flow models are deep equilibrium models (DEQs), which observe that many deep networks converge to fixed points \cite{bai2019deepequilibriummodels}. DEQs are solved directly for equilibrium points using implicit differentiation. The connection to classical system identification arises through the shared emphasis on fixed-point structure and implicit dynamics. Although DEQs do not explicitly model continuous-time dynamics, their fixed-point formulation provides useful parallels to implicit ODE solvers and stable attractor learning.

In contrast to these methods, we propose an alternative approach. We develop a stability-constrained NODE framework for learning asymptotically stable systems that exhibit multi-attractor and hysteretic dynamics. The learned model guarantees trajectory stability throughout the state space and additionally enables tractable control policies that can navigate nontrivial bifurcations, such as tipping points in hysteresis loops.

\section{Preliminaries}
NODEs provide a learning framework for modeling continuous-time dynamical systems by parameterizing the governing dynamics with a neural network. A NODE models the state evolution as

\begin{equation}
    \frac{\mathrm{d}\vect{x}}{\mathrm{d}t} = F_\theta(\vect{x}, t),
    \label{eqn:gen-node}
\end{equation}

where \(\vect{x}(t) \in \mathbb{R}^d\) is the time-dependent state vector and \(F_\theta : \mathbb{R}^d \times \mathbb{R}^+ \to \mathbb{R}^d\) is a neural network parameterized by \(\theta \in \R^p\). Given an initial condition \(\vect{x}(0)\), system trajectories over the interval \(t \in [0, T]\) are obtained by solving the initial value problem
\begin{equation}
    \widehat{\vect{x}}(T) = \vect{x}(0) + \int_0^T F_\theta(\vect{x}(t), t)\, \mathrm{d}t.
    \label{eqn:node-ivp}
\end{equation}
The original formulation in \cite{chen2018neural} uses the \texttt{torchdiffeq} package to numerically integrate \cref{eqn:node-ivp}. However, in this work, we instead use the more recent \texttt{torchode} library \cite{Lienen_torchode_A_Parallel_2022}, which provides a modern, performance-optimized framework to solve batched NODE systems.

The model parameters \(\theta\) are learned by minimizing the discrepancy between the predicted trajectories and the observed states at the measurement times \(t_i\), for \(i =1, \ldots, N\):
\begin{equation}
    \mathcal{L}(\theta) = \frac{1}{N}\sum_{i=1}^{N} \big\| \vect{x}(t_i) - \widehat{\vect{x}}(t_i) \big\|_2^2.
    \label{eqn:NODE-loss}
\end{equation}
Gradients of this loss with respect to \(\theta\) are obtained by backpropagating through the \texttt{torchode} ODE solver using PyTorch's autograd. The solver employs reversible integrators and internal checkpointing, enabling accurate gradients while remaining memory efficient. Unlike the classical adjoint sensitivity method, this approach avoids numerical instabilities in the backward ODE solve and is well suited for large batched problems, making it preferable for our setting.

In addition to the use of the loss \cref{eqn:NODE-loss}, we also test a gradient matching technique for learning with the intent to avoid instabilities that can arise with NODE training objectives as in \cite{okamoto2025learningsimplestneuralode}. We describe both training objectives in \Cref{sec:numerical-results}.

\section{Methodology}

We consider modeling a controlled autonomous ODE system of the form
\begin{equation}
    \frac{\mathrm{d}\vect{x}}{\mathrm{d}t} = F(\vect{x}, \vect{u}),
\end{equation}
where \(\vect{x} \in D \subset \mathbb{R}^d\) denotes the system observables, assumed to be in a closed and convex bounded set \(D\), and \(\vect{u} \in \mathcal{U} \subset \mathbb{R}^q\) denotes control parameters, for some bounded subset \(\mathcal{U}\). We focus on systems whose trajectories remain bounded under bounded inputs and that exhibit only a small number of asymptotically stable equilibria; in particular, for constant inputs we do not expect limit cycles or chaotic attractors.

\subsection{System identification}

We propose that the true dynamics can identically be written as, or closely approximated by,
\begin{equation}
    F(\vect{x}, \vect{u}) \approx f_\theta(\vect{x}) \bigl(\vect{x} - g_\theta(\vect{x}, \vect{u})\bigr),
    \label{eqn:fandg}
\end{equation}
where 
\begin{equation}
    \begin{split}
        f_\theta &: D \to \mathbb{R}^d, 
        \qquad f_\theta(\vect{x}) < 0 \ \text{elementwise for all } \vect{x} \in D,\\[0.2em]
        g_\theta &: D \times \mathbb{R}^q \to D,
    \end{split}
    \label{eqn:fandg2}
\end{equation}

are neural networks parameterized by \(\theta \in \R^p\) and the product in \cref{eqn:fandg} is taken elementwise, so that each component of \(f_\theta(\vect{x})\) acts as a state-dependent decay rate pulling \(\vect{x}\) toward the target \(g_\theta(\vect{x},\vect{u})\). For brevity, we write \(\theta\) to refer collectively to the distinct parameters of \(f_\theta\) and \(g_\theta\); however, they are not necessarily shared.

In particular, the formulation \cref{eqn:fandg} does not impose significant limits on the systems that can be represented. Equality of \cref{eqn:fandg} can be achieved whenever defining \( g_{\theta}(x,u) = x + F(x,u)\) satisfies the condition \cref{eqn:fandg2} on appropriate \(D\), with the choice \(f_{\theta} \equiv -1\). Of course, there is no uniqueness to a splitting of the form \cref{eqn:fandg} and in practice we find many such  splittings can be learned as demonstrated in later numerical examples.

The construction in \cref{eqn:fandg,eqn:fandg2} guarantees trajectory stability across the entire state space \(D\). Previous works obtain stability by constraining a general NODE, of the form \cref{eqn:gen-node}, through a jointly learned Lyapunov function \cite{Kolter2019LearningSD}. However, these approaches are limited to learning a single basin of attraction. In contrast, our formulation is able to represent nontrivial multistable systems, such as hysteresis loops, from limited information and can be naturally extended to feedback control policies.

\subsection{Control}
If a system can be represented in the form \cref{eqn:fandg2} then steady states, denoted \(\vect{x}^*\), satisfy
\[
    \vect{x}^* = g_\theta(\vect{x}^*, \vect{u})
\]
for some \(\vect{u}\). Thus, if a steady state \(\vect{x}^*\) is achievable, stable or unstable, determining an appropriate control amounts to finding a solution to
\begin{equation}
    \argmin_{\vect{u} \in \mathbb{R}^q}
    \frac{1}{2} \bigl\| \vect{x}^* - g_\theta(\vect{x}^*, \vect{u}) \bigr\|_2^2.
    \label{eqn:control-resid}
\end{equation}
For a learned linear control structure, i.e., when \(g_\theta\) does not depend on \(\vect{x}\) and depends linearly on \(\vect{u}\), gradient descent applied to \cref{eqn:control-resid} converges exponentially to the minimum-norm solution (see \Cref{sec:solving-control}).

\subsection{Continuous feedback control}
The formulation above admits a natural extension to a continuous-time feedback control policy. For example, given a current state \(\vect{x}\), we may continuously update \(\vect{u}\) to drive the system toward a desired target \(\vect{x}^*\) with
\begin{equation}
    \frac{\mathrm{d}\vect{u}}{\mathrm{d}t}
    = -\eta \, \nabla_{\vect{u}} 
    \frac{1}{2}
    \bigl\| g_\theta^{\circ k}(\vect{x}, \vect{u}) - \vect{x}^* \bigr\|_2^2,
    \label{eqn:cont-control-dyn}
\end{equation}
where \(\eta > 0\) modulates the update strength, and \(g_\theta^{\circ k}\) denotes the \(k\)-fold iteration
\begin{equation}
    g^{\circ k}(\vect{x},\vect{u}) \triangleq g(\vect{x}_{k},\vect{u}) \quad \text{where} \quad \vect{x}_i = g(\vect{x}_{i-1},\vect{u}) \: \text{for} \: i=1,\hdots,k.
    \label{eqn:g-iter}
\end{equation}

Increasing \(k\) acts as an inexpensive approximate forecaster of the dynamics. There is a trade-off between larger \(k\) and computational efficiency when solving \cref{eqn:cont-control-dyn}, which can be mitigated by leveraging automatic differentiation. In the scalar case we show that increasing \(k\) can improve stability under the additional assumption \(|\frac{\partial}{\partial x} g_\theta(\vect{x}^*,\vect{u})| < 1\) (see \Cref{sec:1d-system}).
 We illustrate the effects of varying \(k\) and \(\eta\) on numerical examples in \Cref{sec:numerical-results}.

\subsection{Constrained feedback control}

The control policy in \cref{eqn:cont-control-dyn} can easily be extended to enforce boundary constraints on the control inputs. Namely, by using the form
\begin{equation}
        \frac{\mathrm{d}\vect{u}}{\mathrm{d}t}
    = -\eta \, \nabla_{\vect{u}} 
    \frac{1}{2}
    \bigl\| g_\theta^{\circ k}(\vect{x}, \vect{u}) - \vect{x}^* \bigr\|_2^2 \odot \phi(\vect{u})
    \label{eqn:constrained-control}
\end{equation}
where \(\odot\) is the Hadamard product and \(\phi(\vect{u}) \in \R^q\) is a vector composed of a linear combination of smooth Heaviside functions 
\begin{equation}
    \phi(u_i) = \sum_j \alpha_j H(u_i - \delta_j, \ell_j)
\end{equation}
where \(H(x,\ell) = 1 / (1 + \exp{(-\ell x)})\), \(\alpha_j = \pm 1\), \(\delta_j\) are the boundary limits for \(u_i\), and \(\ell_j\) are the rates of decay in control.

\section{Theoretical analysis of gradient-based control}
\label{sec:theoretical-analysis}
We analyze the stability properties of the equilibrium map \(g_\theta\)
and the convergence of the gradient-based feedback control policy
introduced in \cref{eqn:cont-control-dyn,eqn:constrained-control}.
We first establish local stability and contraction results for the
one-dimensional case, then examine the gradient flow for a tractable
linear simplification of the control objective \cref{eqn:control-resid}
to build intuition for the numerical experiments in \Cref{sec:numerical-results}.

\subsection{One dimensional systems}

\label{sec:1d-system}

For scalar dynamical systems of the form \cref{eqn:fandg} we provide criteria for exponential stability of equilibria, and stability of the iteration in \cref{eqn:g-iter} for control.
\begin{lemma}[Linear stability condition]\label{lem:linear-stability}
Let \(x^* \in \R\) be a fixed point of the scalar system
\[
    \frac{\mathrm{d}x}{\mathrm{d}t} = f(x)\bigl(x - g(x,u)\bigr),
\]
for a given constant input \(u \in \mathcal{U}\) (i.e. \(x^* = g(x^*,u)\)). Suppose \(f\in C^1(\R)\)  with \(f(x^*)<0\), and 
the partial map \(g(\cdot, u) \in C^1(\R)\).  
Then \(x^*\) is a locally exponentially stable equilibrium if and only if
\[
    \frac{\partial}{\partial x} g(x^*,u) < 1.
\]
\end{lemma}

\begin{proof}
Define \(F(x,u) := f(x)\bigl(x-g(x,u)\bigr)\) then we have
\[
    \frac{\partial}{\partial x} F(x,u) = f'(x)\bigl(x-g(x,u)\bigr) + f(x)\bigl(1-\frac{\partial}{\partial x} g(x,u)\bigr).
\]
At \(x=x^*\) the first term vanishes because \(x^*=g(x^*,u)\); therefore
\[
    \frac{\partial}{\partial x} F(x,u)\big|_{x=x^*} = f(x^*)\bigl(1-\frac{\partial}{\partial x} g(x^*,u)\bigr).
\]
Exponential stability is given by \(\frac{\partial}{\partial x} F(x^*,u)<0\). Since \(f(x^*)<0\), this is equivalent to \(1-\frac{\partial}{\partial x} g(x^*,u)>0\), i.e. \(\frac{\partial}{\partial x} g(x^*,u)<1\). 
\end{proof}

\begin{theorem}[Local contraction of \(g\)]\label{thm:local-contraction}
Under the hypotheses of \cref{lem:linear-stability}, assume in addition that
\[
    \left|\frac{\partial}{\partial x} g(x^*,u)\right| < 1.
\]
Then there exist constants \(L\in(0,1)\) and \(r>0\) such that \(g(\cdot,u)\) is a contraction on the closed ball \(B_r(x^*)\) with Lipschitz constant \(L\), i.e.
\[
    |g(x,u) - g(y,u)| \le L |x-y|, \quad \text{for}\quad x,y\in B_r(x^*),
\]
and moreover \(g(B_r(x^*))\subset B_r(x^*)\).
\end{theorem}

\begin{proof}
Because \(g\) is \(C^1\) in \(x\) and \(|\frac{\partial}{\partial x} g(x^*,u)|<1\), continuity of \(\frac{\partial}{\partial x} g(\cdot,u)\) gives \(r>0\) and \(L\in(0,1)\) such that
\[
    \sup_{x\in B_r(x^*)} \left|\frac{\partial}{\partial x} g(x,u)\right| \le L.
\]
For any \(x,y\in B_r(x^*)\) the mean-value (integral) form yields
\[
    |g(x,u)-g(y,u)| = \left|\int_y^x \frac{\partial}{\partial x} g(s,u)\,\mathrm{d}s\right| \le L|x-y|,
\]
so \(g\) is Lipschitz on \(B_r(x^*)\) with constant \(L<1\). Taking \(y=x^*\) we obtain
\[
    |g(x,u)-x^*| \le L|x-x^*| < Lr < r,
\]
hence \(g(B_r(x^*))\subset B_r(x^*)\).
\end{proof}

\begin{corollary}[Geometric convergence of the iteration \cref{eqn:g-iter}]\label{cor:geometric}
Under the conditions of \cref{thm:local-contraction}, for any \(x_0\in B_r(x^*)\) the sequence
\[
    x_{k+1} = g(x_k,u)
\]
converges to \(x^*\) and satisfies the geometric bound
\[
    |x_k - x^*| \le L^k |x_0 - x^*|.
\]
\end{corollary}

\begin{proof}
By \cref{thm:local-contraction} the map \(g(\cdot,u)\) is a contraction on the complete metric space \(B_r(x^*)\) with unique fixed point \(x^*\). The Banach fixed-point theorem yields the stated convergence and bound.
\end{proof}

\begin{remark}[Pointwise derivative vs.\ uniform contraction]
The condition \(\|\frac{\partial}{\partial x} g(x^*,u)\| < 1\) is a \emph{pointwise} requirement ensuring that the
linearization of \(g\) at the fixed point is locally stable.  
This guarantees that the ODE 
\[
    \frac{\mathrm{d}x}{\mathrm{d}t} = f(x)\bigl(x-g(x,u)\bigr)
\]
has a stable equilibrium at \(x^*\).

However, this condition alone does \emph{not} imply that \(g\) is a contraction on any neighborhood of
\(x^*\).  The Banach fixed-point theorem requires a \emph{uniform} bound: there must exist a ball
\(B_r(x^*)\) and a constant \(L<1\) such that
\[
    \sup_{x\in B_r(x^*)} \left\|\frac{\partial}{\partial x} g(x,u)\right\| \le L < 1.
\]
This ensures the global Lipschitz inequality
\[
    \|g(x,u)-g(y,u)\| \le L \|x-y\|,\quad \text{for} \quad x,y\in B_r(x^*),
\]
which is strictly stronger than the pointwise condition at \(x^*\).

Thus, while \(\|\frac{\partial}{\partial x} g(x^*,u)\| < 1\) guarantees linear stability of the ODE equilibrium, it may
still happen that
\[
    \sup_{x \in B_r(x^*)} \left\|\frac{\partial}{\partial x} g(x,u)\right\| \;\ge\; 1
    \qquad\text{for every } r>0,
\]
in which case \(g\) fails to be a contraction on any neighborhood, and the discrete iteration
\(x_{k+1}=g(x_k,u)\) need not converge. Thus, an increased number of iterations does not necessarily improve stability automatically when solving for control.
\end{remark}

\subsection{Gradient descent} 

\label{sec:solving-control}
Having established local stability and contraction of \(g_\theta\) 
in \Cref{sec:1d-system}, we now analyze the control objective 
introduced in \cref{eqn:control-resid}. To build intuition for the use of continuous-time feedback policies in the later numerical examples, \Cref{sec:numerical-results}, we examine an analytically tractable simplification where \(g\) is constant with respect to \(\vect{x}\), and linear with respect to \(\vect{u}\), i.e., \(g(\vect{x},\vect{u}) = G\vect{u}\) for a matrix \(G\). The objective \cref{eqn:control-resid} becomes
\begin{equation}
    \mathcal{L}(\vect{u}) = \frac{1}{2} \|G\vect{u} - \vect{x}^*\|_2^2,
    \label{eqn:linear-obj}
\end{equation}
where \(\vect{x}^*\) is the desired state. If \(G\) is full rank, the minimizer is
\begin{equation}
    \vect{u}^\star = (G^\intercal G)^{-1} G^\intercal\vect{x}^*.
\end{equation}
If \(G\) is rank-deficient, standard approaches include using the Moore-Penrose pseudoinverse,
\begin{equation}
    G^\dagger = V \Sigma^\dagger U^\intercal,
\end{equation}
where \(\Sigma^\dagger\) contains the reciprocals of the nonzero singular values, or adding a small ``nugget'' \(\lambda\) to handle singularity:
\begin{equation}
    \vect{u}^\star = (G^\intercal G + \lambda I)^{-1} G^\intercal \vect{x}^*.
\end{equation}

A minimizer of \eqref{eqn:linear-obj} over \(\vect{u} \in \mathcal{U}\) can be solved via gradient descent:
\begin{equation}
    \vect{u}_{k+1} = \vect{u}_k - \eta \nabla \mathcal{L}(\vect{u}_k) = \vect{u}_k - \eta G^\intercal (G \vect{u}_k - \vect{x}^*),
\end{equation}
which is a linear iteration for the quadratic objective
\begin{equation}
    \mathcal{L}(\vect{u}) = \frac{1}{2} \vect{u}^\intercal H \vect{u} + \vect{b}^\intercal \vect{u} + \frac{1}{2} (\vect{x}^*)^\intercal \vect{x}^*, \quad H = G^\intercal G, \quad \vect{b} = -G^\intercal \vect{x}^*.
\end{equation}
Let \(\sigma_{\max}(G)\) and \(\sigma_{\min}(G)\) denote the largest and smallest nonzero singular values of \(G\). Then the gradient is \(L\)-Lipschitz with
\begin{equation}
    L = \sigma_{\max}(G)^2,
\end{equation}
and if \(G\) is full rank, the objective is strongly convex with
\begin{equation}
    \mu = \sigma_{\min}(G)^2 > 0,
\end{equation}
giving condition number \(\kappa(H) = L / \mu = \kappa(G)^2\). 

For gradient descent with step size \(\eta^\star = 2/(L+\mu)\),
\begin{equation}
    \|\vect{u}_k - \vect{u}^\star\| \le \rho(\eta^\star)^k \|\vect{u}_0 - \vect{u}^\star\|, \quad \rho(\eta^\star) = \frac{\kappa(G)^2 - 1}{\kappa(G)^2 + 1}.
\end{equation}
Thus, if \(\kappa(G) \sim 1\), \(\rho \sim 0\) yielding rapid convergence. To achieve a tolerance \(\epsilon\),
\begin{equation}
    k = \log_{1/\rho} \frac{\|\vect{u}_0 - \vect{u}^\star\|}{\epsilon}.
\end{equation}

\subsection{Continuous-time gradient flow} 
Assuming \(G\) is full rank, feedback control can be updated in continuous time:
\begin{equation}
\begin{split}
    \frac{\mathrm{d}\vect{u}}{\mathrm{d}t} &= -\frac{\eta}{2} \nabla \|G\vect{u} - \vect{x}^*\|_2^2 = -\eta G^\intercal(G\vect{u} - \vect{x}^*) \\
           &= -\eta H(\vect{u} - \vect{u}^\star),
\end{split}
\end{equation}
with error \(\vect{e}(t) = \vect{u}(t) - \vect{u}^\star\) evolving as
\begin{equation}
    \frac{\mathrm{d}\vect{e}}{\mathrm{d}t} = -\eta H \vect{e} \quad \Rightarrow \quad \vect{e}(t) = e^{-\eta H t} \vect{e}(0),
\end{equation}
giving
\begin{equation}
    \|\vect{u}(t) - \vect{u}^\star\|_2 = \|e^{-\eta H t} (\vect{u}(0) - \vect{u}^\star)\| \le e^{-\eta \sigma_{\min}(G)^2 t} \|\vect{u}(0) - \vect{u}^\star\|,
\end{equation}
i.e., exponential convergence. Gradient descent is simply an explicit discretization of this continuous-time flow.

\subsection{Rank-deficient case} 
If \(G\) is rank-deficient, adding a regularization term \( \frac{\lambda}{2} \|\vect{u}\|_2^2 \) ensures a non-singular Hessian:
\begin{equation}
    \nabla \mathcal{L}(\vect{u}) = G^\intercal (G\vect{u} - \vect{x}^*) + \lambda \vect{u}, \quad \nabla^2 \mathcal{L}(\vect{u}) = G^\intercal G + \lambda I.
\end{equation}
This is equivalent to Tikhonov (ridge) regularization.

\section{Numerical results}
\label{sec:numerical-results}
We test our method on several nonlinear examples, demonstrating reliable short-horizon learning and accurate, noise-robust feedback control across systems that exhibit nontrivial bifurcations, and demonstrating the ease of control constraint implementation. All experiments and code are made publicly available at \cite{pnnl_ftnode}.

\subsection*{Metrics}
We report a normalized root mean square error (nRMSE) for feedback control experiments where
\begin{equation}
    \text{nRMSE} = \frac{\sqrt{\frac{1}{n}\sum_{i=1}^n \big(x(t_i) - x^*\big)^2}}{\text{system magnitude}},
\end{equation}
is evaluated at the end of an allotted time to reach a randomized target. Systems are formulated as stochastic differential equations (SDEs) and are solved using a Euler--Maruyama solver with randomized, pregenerated, noise. 

The system magnitude is defined as the range of trajectory outputs when these values are easily known, and otherwise is defined as the interquartile range (IQR), defined as the difference between the 75th and 25th percentiles of all observed trajectories. For multidimensional systems, nRMSE is computed and reported independently for each dimension.

\subsection*{Model architecture and cross validation}
For all examples, we use multi-layer perceptrons (MLPs) for both \(f_\theta\) and \(g_\theta\), equipped with Sigmoid-weighted Linear Unit (\texttt{SiLU}) activation functions and sigmoidal output layers. The specific bounds and layer widths are specified for each example.

We perform a \(10\)-fold cross validation across various model architectures to select a candidate model for learning based on the best average validation score. Due to the inherent nature of the problem, specifically dynamics learning with delicate bifurcation information, we do a final training on all available data, and select the best performing model. This is standard in scientific and medical contexts where scarce data are present \cite{Tsamardinos2022, Bradshaw2023}.  We note that we avoid the typical pitfall of highly oscillatory overfitting due to asymptotically enforced stability of trajectories in the model design.

\subsection*{Training}
 We use an Adam optimizer \cite{Kingma2014AdamMA} equipped with a \texttt{ReduceLROnPlateau} scheduler, with specific hyperparameters, such as the number of training epochs, batch size, and initial learning rate specified per example. 

 We use two training objectives given in \cref{eqn:train-obj} that we will refer to as \textit{Trajectory matching} and \textit{Gradient matching}. Trajectory matching falls within the standard context of NODE training: using backpropagation to obtain simulated trajectories from a model to compare against true observed trajectories. For gradient matching, we instead compare the learned model against gradient information from the observed trajectory, computed using centered finite difference with forward and backward Euler at the respective boundaries. That is, given a set of observables, \(\{\vect{x}_i\}_{i=1}^n\), we define our tunable model as 
 
\begin{equation}
\begin{split}
    F_{\theta}(\vect{x},\vect{u}) &= f_{\theta}(\vect{x})(\vect{x}-g_\theta(\vect{x},\vect{u}))\\
    \widehat{\vect{x}}_\theta(t) &=\texttt{ODESolve}(F_\theta, \vect{x}_0, \vect{u}_0, t)\\
    \end{split}
\end{equation}
and the finite difference scheme as 
\begin{equation}
D(\vect{x}_i) =
\begin{cases}
    \frac{\vect{x}_{i+1}-\vect{x}_{i-1}}{t_{i+1}-t_{i-1}}, & i = 2, \cdots, n-1\\
    \frac{\vect{x}_2-\vect{x}_1}{t_2 - t_1}, & i=1\\
    \frac{\vect{x}_n - \vect{x}_{n-1}}{t_{n}-t_{n-1}}, & i=n,
    
\end{cases}
\end{equation}
where we note that our observations \(t_i\) are uniform. We define our two training objectives as
\begin{equation}
\begin{split}
    \mathcal{L}_{\text{Traj}}(\theta) &\triangleq \frac{1}{n}\sum_{i=1}^n\|\vect{x}_i - \widehat{\vect{x}}_\theta(t_i)\|_2^2\\
    \mathcal{L}_{\text{Grad}}(\theta)&\triangleq \frac{1}{n}\sum_{i=1 }^n \|D(\vect{x}_i)-F_\theta(\vect{x}_i, \vect{u}_i)\|_2^2
\end{split}
\label{eqn:train-obj}
\end{equation}
where the appropriate training objective is determined through preliminary empirical evaluation.

\subsection{Mixing tanks}
\label{sec:mixing-tanks}
We consider a pair of connected tanks controlled by a single pump \(p\) and a two-way valve \(v\), which allocates the inflow between the two tanks. This setup provides a simplified model of pumped-storage hydroelectricity used for load balancing \cite{DOE_PumpedStorageHydropower}. The system dynamics are given by 
\begin{equation}
    \begin{split}
        \frac{\mathrm{d} x_1}{\mathrm{d}t} &= c_1^{\text{in}}(1-v)p - c_1^{\text{out}} \sqrt{x_1},\\
        \frac{\mathrm{d} x_2}{\mathrm{d}t} &= c_2^{\text{in}} v p + c_1^{\text{out}} \sqrt{x_1} - c_2^{\text{out}} \sqrt{x_2},
    \end{split}
    \label{eqn:two-tanks}
\end{equation}
where \(x_1\) and \(x_2\) denote the tank levels. The flow-rate coefficients are defined as
\begin{equation}
    \begin{split}
        c_1^{\text{in}} &\triangleq a_1\bigl(1 - \sigma_\ell(x_1 - 1)\bigr),\\
        c_1^{\text{out}} &\triangleq a_2\bigl(1 - \sigma_\ell(x_2 - 1)\bigr),\\
         c_2^{\text{in}} &\triangleq a_3\bigl(1 - \sigma_\ell(x_2 - 1)\bigr),\\
        c_2^{\text{out}} &\triangleq a_4,
    \end{split}
    \label{eqn:flow-rate-coeff}
\end{equation}
with \(\sigma_\ell(x) = 1/(1+\exp(-\ell x))\) to maintain physically meaningful dynamics.  We let \(a_1 = a_3= 0.08\), \(a_2 = a_4 = 0.02\) for the flow rate coefficients and \(\ell = 50\) be the boundary rate decay. The parameters \(a_1\), \(a_2\), and \(a_3\) correspond to the unconstrained flow coefficients, while \(c_2^{\text{out}}\) remains constant---this can be interpreted as tank 2 draining into an unconstrained reservoir.

For model training, we simulate trajectories of \cref{eqn:two-tanks} up to a fixed \(t=200\) with initial conditions \(\vect{x}(0) = (x_0, x_0)\) where \(x_0 = 0.05 \cdot i\) for \(i=0, \hdots, 20\) and control configurations \(p,v = 0.1 \cdot i\) for \(i=1, \hdots, 9\). The experimental setup can be seen in \Cref{fig:two-tanks-exp}.

\subsubsection*{System identification} 

Using gradient matching \(\mathcal{L}_{\text{Grad}}\) \cref{eqn:train-obj}, we perform a 10-fold cross validation for 1000 epochs, batch size of 50, and an initial learning rate \(\eta =0.01\). The selected model was a 2-20-20-20-2 and 4-20-20-20-2 fully connected MLP with SiLU activation functions for \(f_\theta\) and \(g_\theta\), respectively, achieving an average best validation loss of \(2.457 \times 10^{-6}\) across all folds. The model is then tuned on all the data with the same training configuration with the best model achieving a loss of \(4.247 \times 10^{-7}\). Model outputs are \((-1,0)\) and \((0,1)\) for \(f_\theta\) and \(g_\theta\), respectively. The inputs to \(g_\theta\) are the concatenated states \(\vect{x} = (x_1,x_2)\) and controls \(\vect{u} = (p,v)\). 

Model-simulated trajectories using \(f_\theta(\vect{x})\bigl(\vect{x} - g_\theta(\vect{x},\vect{u})\bigr)\) are shown in \Cref{fig:mixing-tanks-pp} while the predicted steady states obtained directly from \(g_\theta\) by solving the root-finding problem \(r(\vect{x}) = \vect{x} - g_\theta(\vect{x},\vect{u})\) are shown in \Cref{fig:g-err}.

\begin{figure}
    \centering

    \begin{minipage}{0.9\linewidth}
        \centering
        \includegraphics[width=\linewidth]{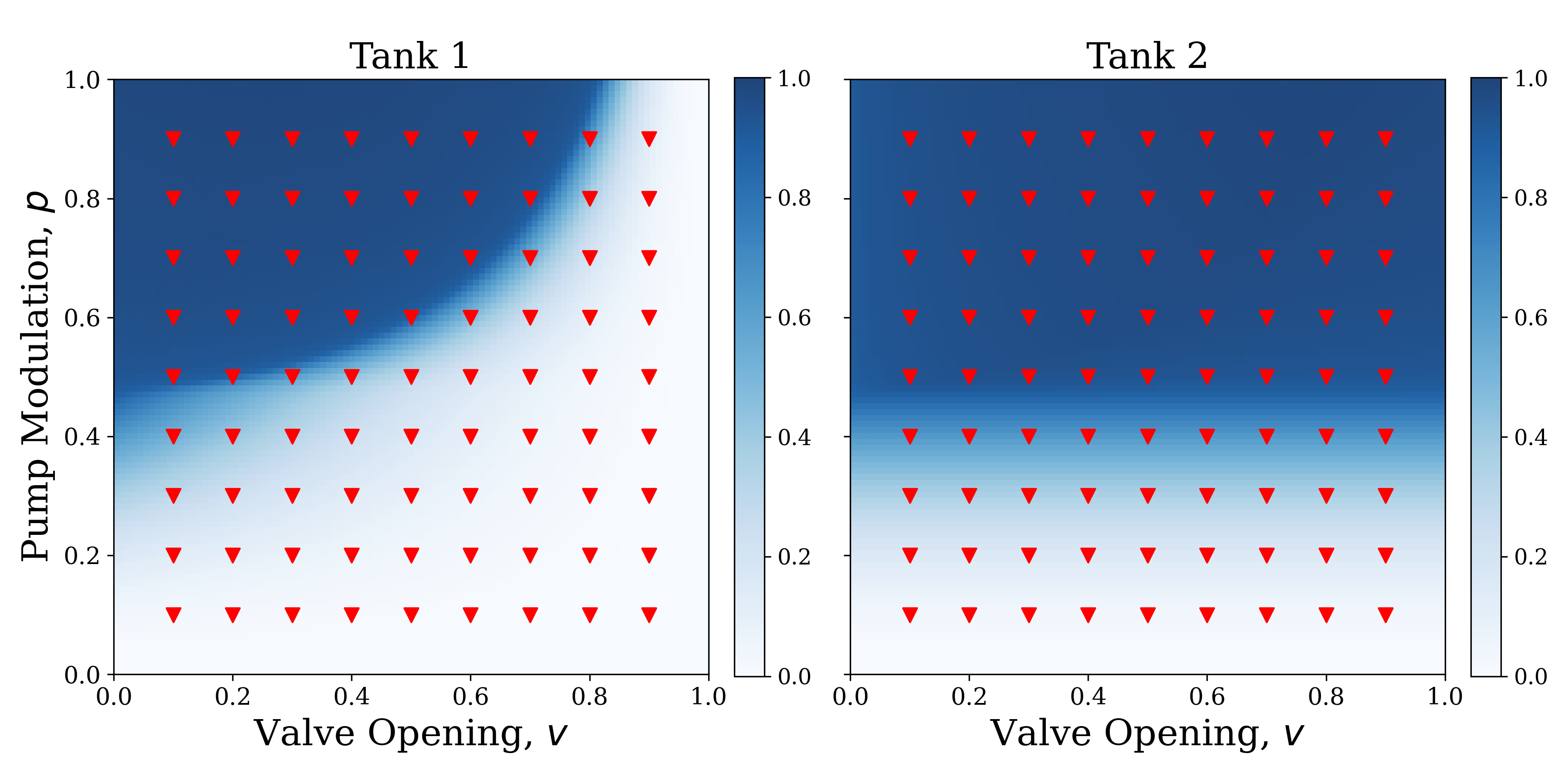}
    \end{minipage}

    \vspace{0pt}

    \begin{minipage}{0.9\linewidth}
        \centering
        \includegraphics[width=\linewidth]{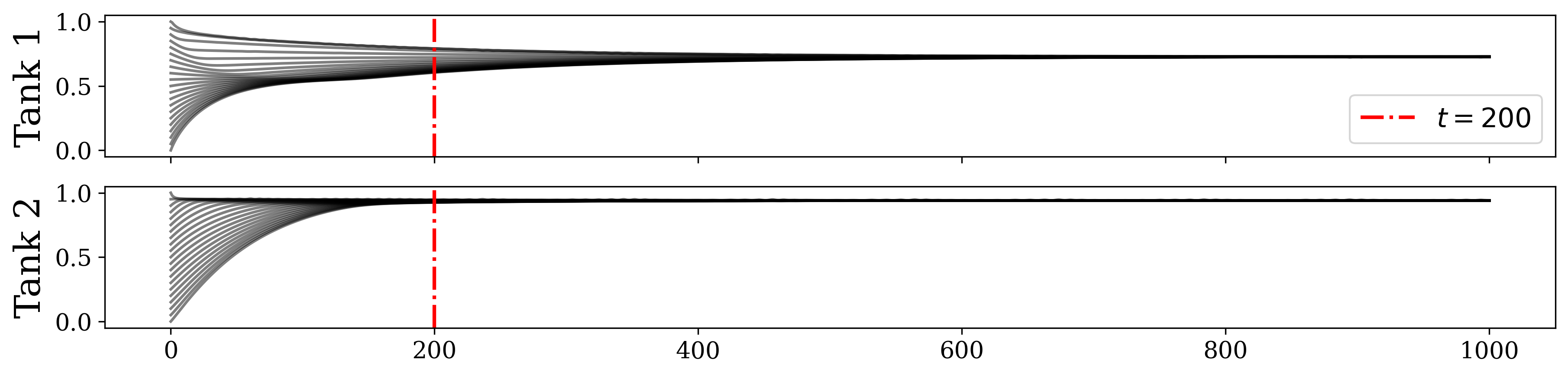}
    \end{minipage}

    \caption{\textbf{Top:} Simulated steady states at \(t = 1000\) for control settings \(v,p \in [0,1]\). Red triangles indicate training configurations \(v,p = 0.1 + 0.1 i\) for \(i = 0,\dots,8\).  \textbf{Bottom:} 21 training trajectories for the configuration \(p =  0.5\) and \(v = 0.25\). Vertical line denotes training cutoff, \(t=200\).}
    \label{fig:two-tanks-exp}
\end{figure}

\begin{figure}[t]
    \centering

    \includegraphics[width=0.7\linewidth]{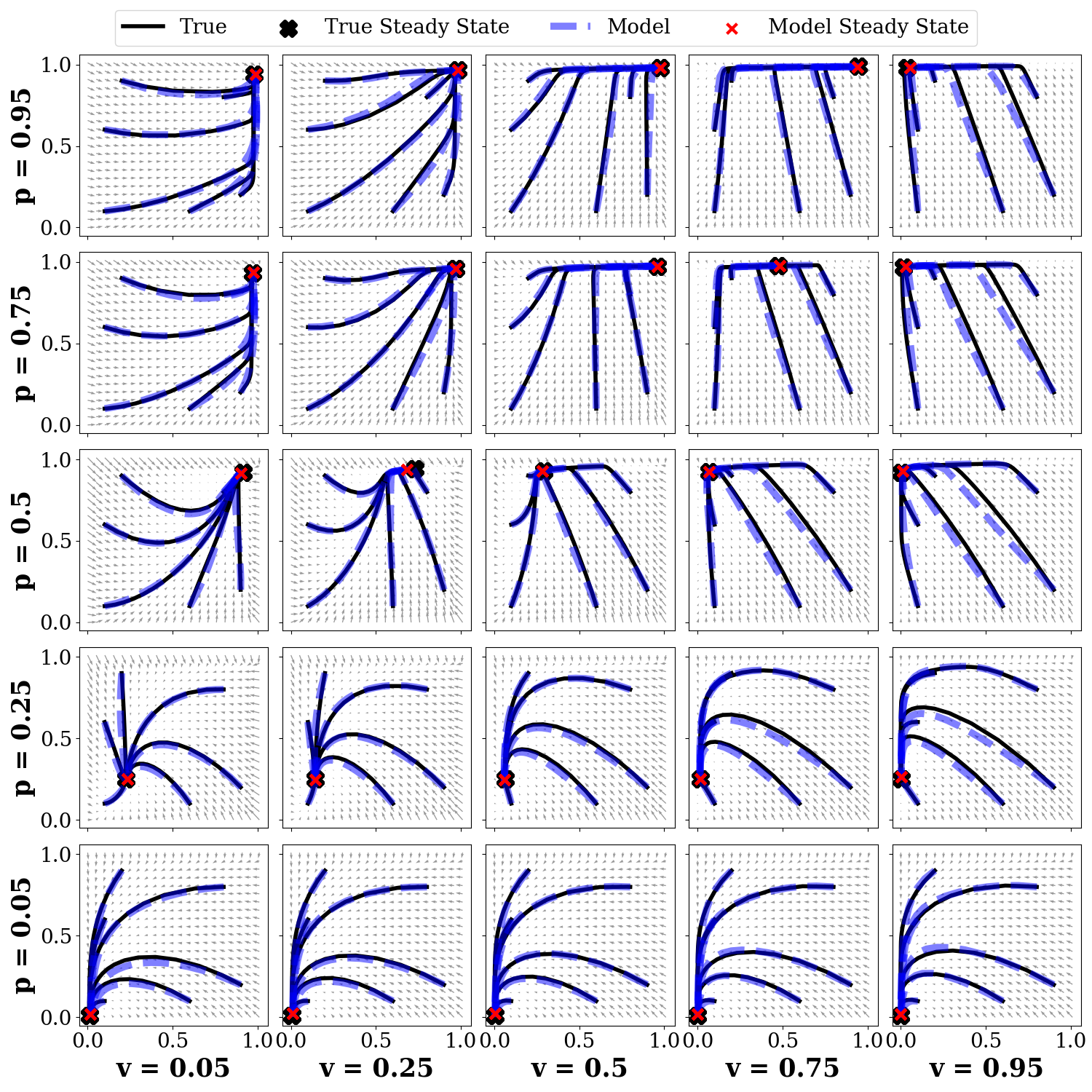}
    \caption{Six trajectories with initial conditions \(\vect{x}(0)\) = (\(0.1\), \(0.1\)),
    (\(0.9\), \(0.2\)),
    (\(0.2\), \(0.9\)),
    (\(0.1\), \(0.6\)),
    (\(0.6\), \(0.1\)), and  (\(0.8\), \(0.8\))
    for 25 control configurations plotted on true phase portraits generated  from  \cref{eqn:two-tanks} (black line) compared against model simulated trajectories from \(f_\theta(\vect{x})(\vect{x}-g_\theta(\vect{x}, \vect{u}))\) (blue dashed line).  True and model steady states are denoted by black and red crosses, respectively.}
    \label{fig:mixing-tanks-pp}
\end{figure}

\begin{figure}
    \centering
    \includegraphics[width=0.9\linewidth]{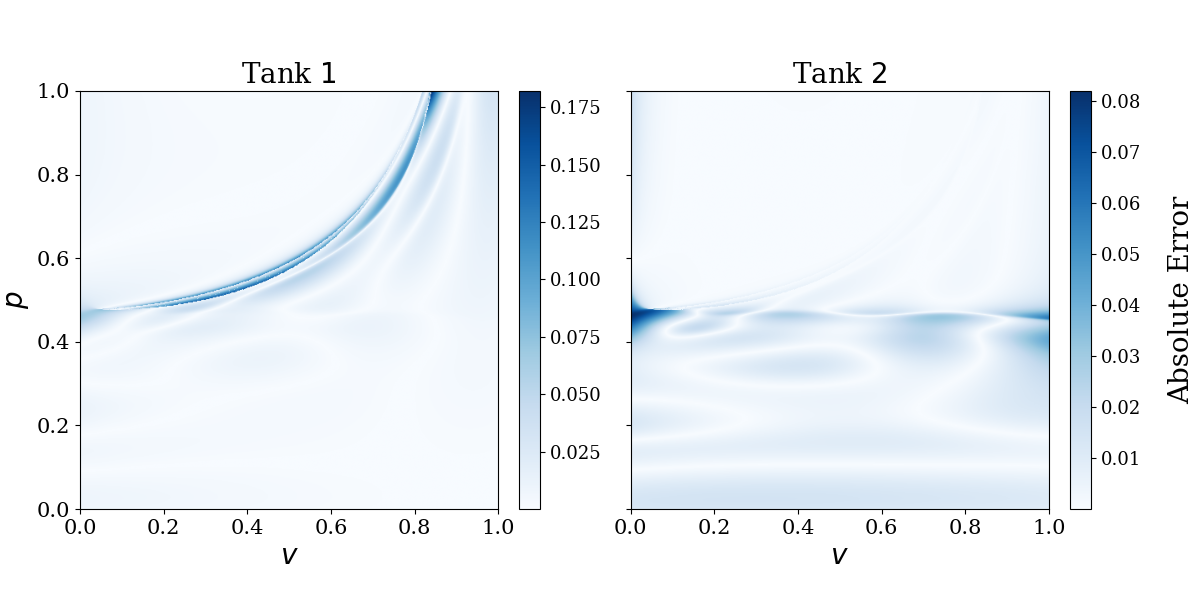}
    \caption{ Absolute error, for tanks 1 and 2, between model and true steady states for control values \(p,v \in [0,1]\). Model steady states are obtained from solving associated root finding problem, \(\vect{x} = g_\theta(\vect{x}, \vect{u})\). 
    True steady states obtained from solving system numerically at \(t=1000\).}
    \label{fig:g-err}
\end{figure}

\subsubsection*{Feedback control}
We conducted \(100\) trials in which the system was steered toward 10 randomized targets, perturbed by state-dependent multiplicative noise, using the constrained feedback control policy given in \cref{eqn:constrained-control}. The targets were generated by simulating the true dynamics under randomly sampled control configurations. Noise is added only to the tank states with covariance given by
\[
G(\vect{x}) = \mathrm{diag}\!\left(\sigma \sqrt{|x_1|},\, \sigma \sqrt{|x_2|}\right),
\]
so that the noise amplitude scales with the water level, while the control dynamics remain noise-free.

The control constraint, for \(\vect{u} = (p,v)\), is given by 
\begin{equation}
    \phi(u_i) =   H(u_i - \delta_1, \ell) -H(u_i - \delta_2, \ell) \in \R^2, \quad u_i = p,v
    \label{eqn:tank-control-constraint}
\end{equation}

where \(\delta_1= .05\), \(\delta_2 = 0.95\), and \(\ell=50\), enforcing \(p,v \in (0,1)\) maintaining meaningful physical interpretation.

We report the average nRMSE computed over the final \(20\%\) of the allotted time for reaching each target, namely \(t=500\). The average nRMSE across all steady states was \(1.54 \times 10^{-2} \) with standard deviation \( 1.38 \times 10^{-2}\) for Tank 1 and \(1.82 \times 10^{-2}\) with a standard deviation \(1.33 \times 10^{-2}\) for Tank 2. Additionally, 97\% and 98\% of steady-state values fell within 5\% of their target for Tanks 1 and 2, respectively, while 78\% and 70\% fell within 2\%. An example feedback control experiment is shown in \Cref{fig:ref-track-tanks}.

\begin{figure}
    \centering
    \includegraphics[width=1\linewidth]{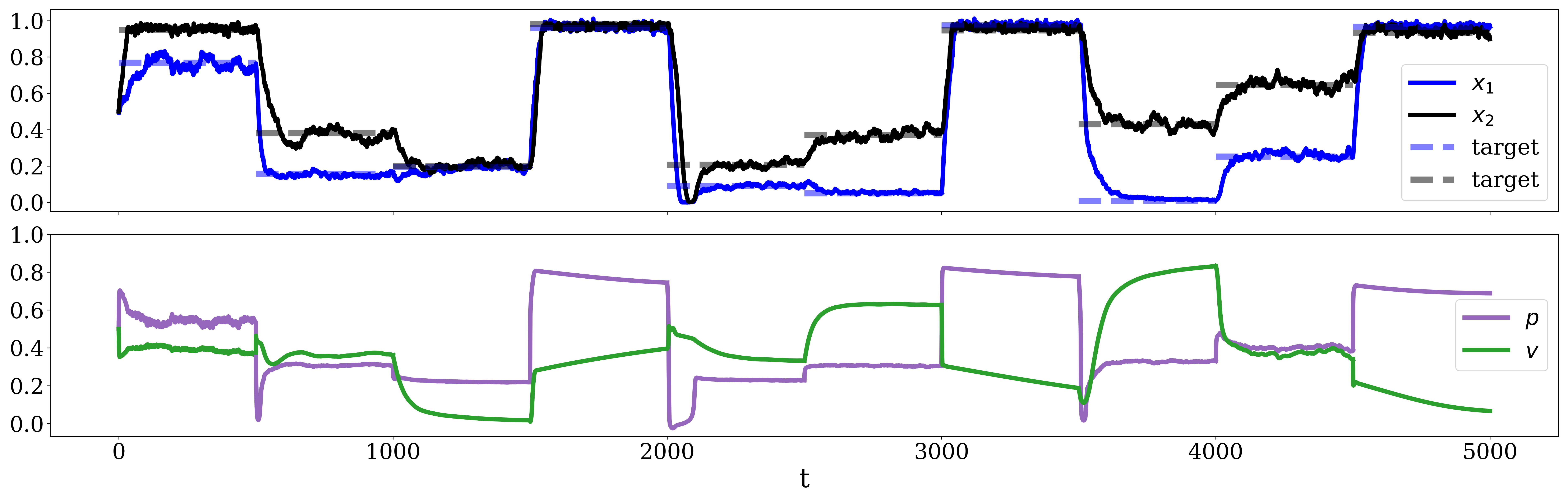}
    \caption{Stochastic feedback control experiments with \(10\) randomized targets reassigned every \(t=500\). \textbf{Top:} Itô simulations with state-proportional noise of scale \(\sigma=0.01\).  \(g_\theta\) for both tanks 1 and 2 plotted, where control iteration and strength are \(k=10\) and \(\eta=0.1\), respectively. \textbf{Bottom:} Constrained continuous feedback control solving \cref{eqn:constrained-control} jointly with the stochastic system where control constraint is given by \cref{eqn:tank-control-constraint}.}
    \label{fig:ref-track-tanks}
\end{figure}

\subsection{Symmetric hysteresis} \label{subsec:sym-hys}

We consider two examples of systems with multiple steady states that exhibit hysteresis, \Cref{subsec:sym-hys,sec:budworm}, a bifurcation characterized by irreversible changes, in the sense they cannot be immediately undone, when control parameters cross critical thresholds, often referred to as tipping points. Our goal is to learn a simple representation of the form in \cref{eqn:fandg} that captures these bifurcations and learns a tractable feedback control policy that can navigate the hysteresis loop.

We consider the ODE
\begin{equation}
    \frac{\mathrm{d}x}{\mathrm{d}t} = \lambda + x - x^3,
    \label{eqn:sym-hyst}
\end{equation}
where \(\lambda\) serves as the control parameter. Bistability occurs for 
\(\lambda \in [-2/\sqrt{27},\, 2/\sqrt{27}]\), which correspond to the system’s tipping points.

This system can be written in a splitting as in \cref{eqn:fandg}, with
\begin{equation}
    \begin{split}
        f(x) &= -x^2, \qquad x \neq 0,\\
        g(x,\lambda) &= \frac{1}{x^2}(x+\lambda).
    \end{split}
    \label{eqn:f-g-decomp}
\end{equation}
Although this decomposition satisfies \(f(x) < 0\) and ensures that 
\(x - g(x,\lambda)\) identifies fixed points, it does not satisfy \cref{eqn:fandg}: 
\(g(x,\lambda)\) exhibits a finite–time blow-up as \(x\) approaches 0, ( which is then exactly canceled by \(f(x) = -x^2\) ). 
However, we find that an approximation of this splitting can be learned with \(g_\theta \approx g\) that captures the asymptotic structure, locates the bistable regime, and reproduces the resulting hysteretic dynamics.

For training, we simulate trajectories using 51 uniformly sampled initial conditions \(x \in [-2,2]\) paired with 51 uniformly sampled control values \(\lambda \in [-1,1]\). Training is restricted to a short time horizon of \(t=0.25\), where full temporal evolution occurs on a much higher order of magnitude, providing only extremely limited transient information.

\subsubsection*{System identification}
Using a trajectory matching loss, \(\mathcal{L}_{\text{Traj}}\) \cref{eqn:train-obj}, we perform a 10-fold cross validation over 200 epochs, with batch size of 50, and an initial learning rate \(\eta = 0.01\). 
To aid learning on the scarce transient dynamics, we featurize the observed state using the orthogonal state transformations,
\begin{equation}
    \vect{x}
    = \bigl(x,\,
    \cos(1^2 \pi (x-a)/(b-a)),\,\ldots,\,
    \cos(4^2 \pi (x-a)/(b-a))\bigr)
    \in \mathbb{R}^5,
    \label{eqn:feat-x}
\end{equation}
with \(a=-1.5\) and \(b=1.5\). These five features form the input to the first five channels of \(g_\theta\), and the sixth input is reserved for the control parameter \(\lambda\). The selected models for \(f_\theta\) and \(g_\theta\) were fully connected 1-20-20-1 and 6-20-20-1 MLPs equipped with SiLU activation functions, and output bounds \((-4,-0.1)\) and \((-2,2)\). The model was selected by a 10-fold cross validation search achieving an average best validation loss of \(6.237 
\times 10^{-8}\). The final model was then tuned on all available data, achieving a best loss of \(7.188\times 10^{-8}\).

The learned model produces accurate simulated trajectories, with an average nRMSE of \(1.210 \times 10^{-3}\) across all training trajectories up to \(t=100\), and reliably identifies regions of the state–control space containing multiple steady states; see \Cref{fig:hyst-traj}.

\begin{figure}
    \centering
    \includegraphics[width=0.9\linewidth]{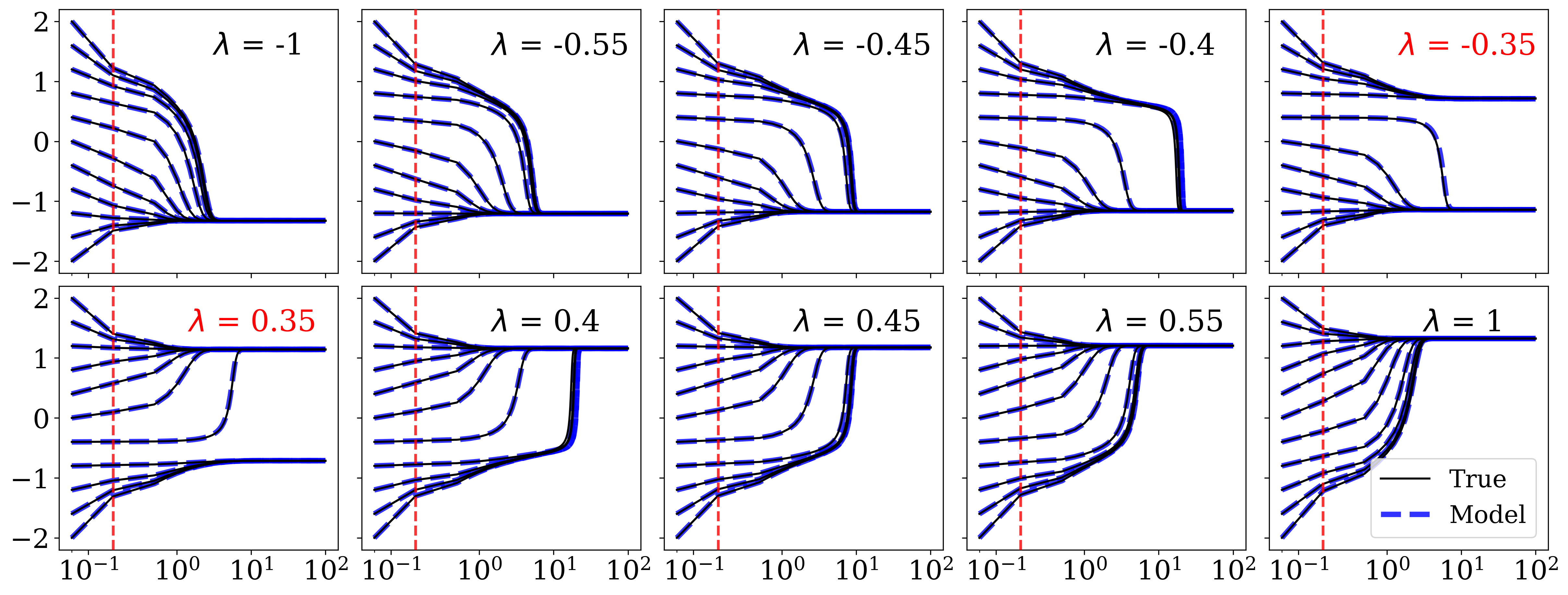}
    \caption{Simulated trajectories (with \(t\) on a logscale) for model \(f_\theta(x)(x-g_\theta(x,\lambda))\) (blue dashed line) trained on data up to \(t=0.25\) (red dashed line) and true dynamics given by \cref{eqn:sym-hyst} (black solid lines) for control values \(\lambda = [-1,-0.55,-0.45,-0.4,-0.35,0.35,0.4,0.45,0.55,1]\) up to \(t=100\) with bistable cases highlighted in red.} 
    \label{fig:hyst-traj}
\end{figure}

The model learns a smooth and bounded \(g_\theta\) and \(f_\theta\) that captures the correct fixed points of the system, which can be seen in \Cref{fig:f-g-hyst}. Additionally, using just \(g_\theta\) we correctly identify the zero-level sets of \(r(x) = x-g_\theta(x,u)\), or in other words, the fixed points of the system in \cref{eqn:sym-hyst}; see \Cref{fig:zero-set-hyst}. This gives rise to a simple loss landscape when solving for control. In fact, this corresponds directly to the bifurcation diagram and is recovered identically in \Cref{fig:hyst-bifur}.

\subsubsection*{Feedback control}

We evaluate the learned \(g_\theta\) as an unconstrained feedback control policy, using \cref{eqn:cont-control-dyn}, for 10 randomized targets uniformly drawn from \([-1.5,1.5]\), over 100 trials, with multiplicative noise added to the states. 

The average nRMSE over the final \(20\%\) of the allotted time, \(t=10\), is \(5.298 \times 10^{-3}\) with standard deviation of \(2.449 \times 10 ^{-3}\). \(100\%\) of all steady states fell within \(2\%\) of their target, with \(94\%\) falling within \(1\%\) of the target.

We highlight the ability of the control policy to identify and counteract tipping point behaviors in an example experiment shown in \Cref{fig:sym-hyst-ref}.

Controlling the dynamics toward an \emph{unstable} steady state is inherently challenging and requires careful tuning of the update strength \(\eta\). We demonstrate the consequences of using an underdamped controller, i.e. \(\eta\) too small, which fails to counteract the tipping behavior of the system in \Cref{fig:eta-effects}.

\begin{figure}
    \centering
    \includegraphics[width=1\linewidth]{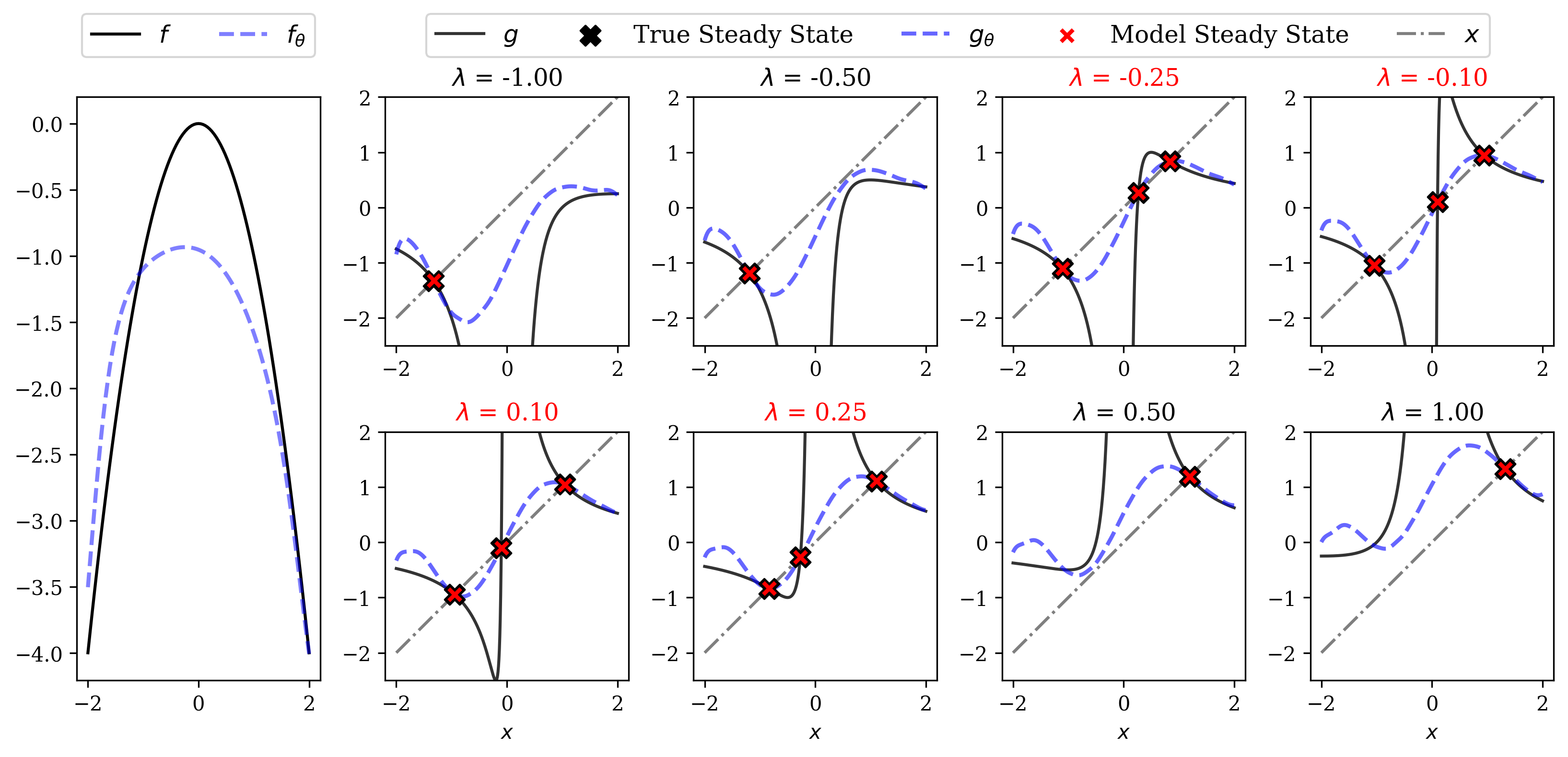}
    \caption{\textbf{Left:} Learned \(f_\theta\) compared against \(f\) from \cref{eqn:f-g-decomp} across training regime \(x\in[-2,2].\) \textbf{Right:} Learned \(g_\theta\) compared against \(g\) in \cref{eqn:f-g-decomp} for various control values \(\lambda\). Red highlighted control indicates bistable regime, \(\lambda \in [-2/\sqrt{27}, 2 / \sqrt{27}]\), that is 3 steady states occur, 2 stable and 1 unstable--denoted by the intersections with the diagonal gray dashed line \(x=x\). }
    \label{fig:f-g-hyst}
\end{figure}

\begin{figure}
    \centering
    \includegraphics[width=1\linewidth]{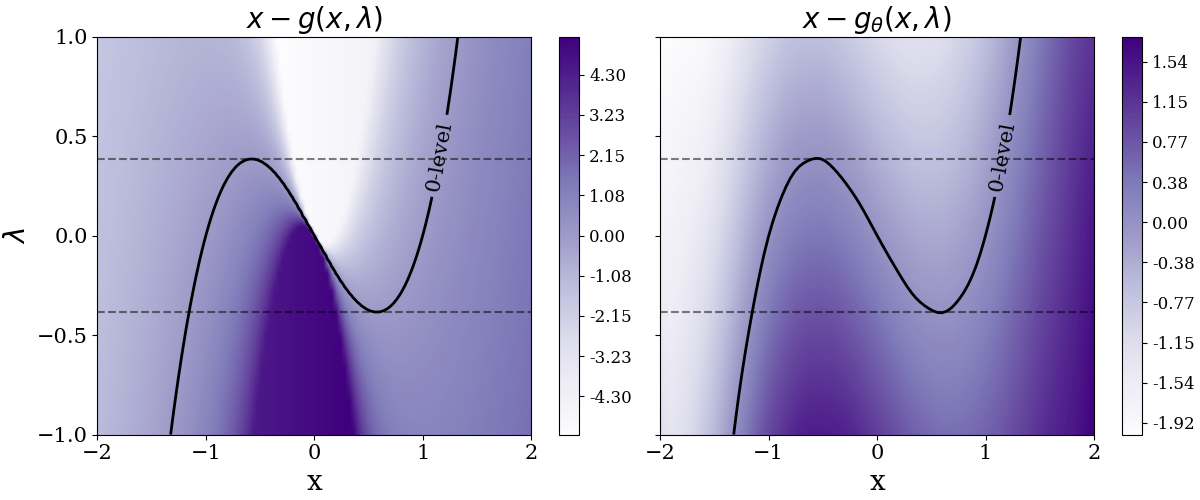}
    \caption{Comparison of the true and learned function outputs across the state and control space. \textbf{Left:} Output of \(x-g(x,\lambda)\) (from \cref{eqn:f-g-decomp}). \textbf{Right:} Output of the learned approximation \(x-g_\theta(x,\lambda)\). In both panels, the zero-level set is denoted by a solid black line, and the black dashed lines indicate the system tipping points at \(\lambda = \pm 2/\sqrt{27}\).}
    \label{fig:zero-set-hyst}
\end{figure}

\begin{figure}
    \centering
    \includegraphics[width=0.7\linewidth]{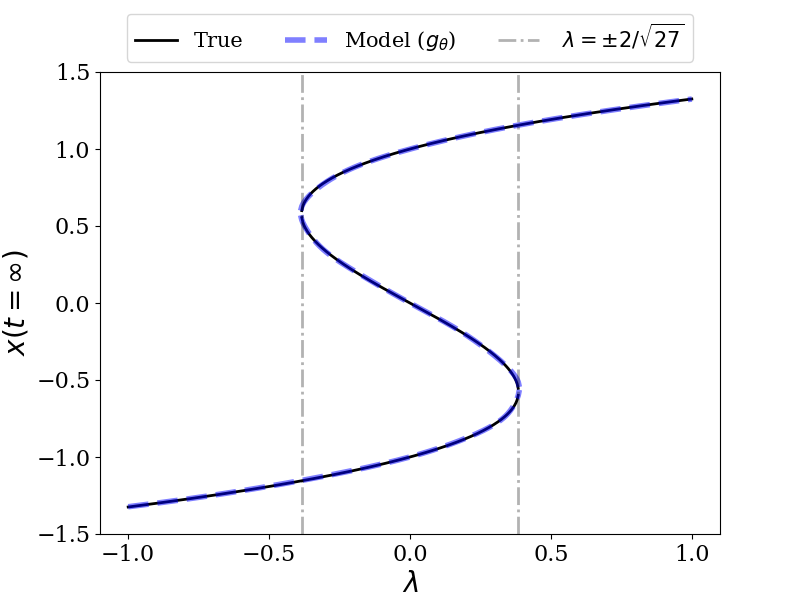}
    \caption{All steady states learned from solving root finding problem \(r(x)=x-g_\theta(x,\lambda)\) for \(\lambda \in [-1,1]\). Gray vertical lines indicate system tipping points.}
    \label{fig:hyst-bifur}
\end{figure}

\begin{figure}
    \centering
    \includegraphics[width=1\linewidth]{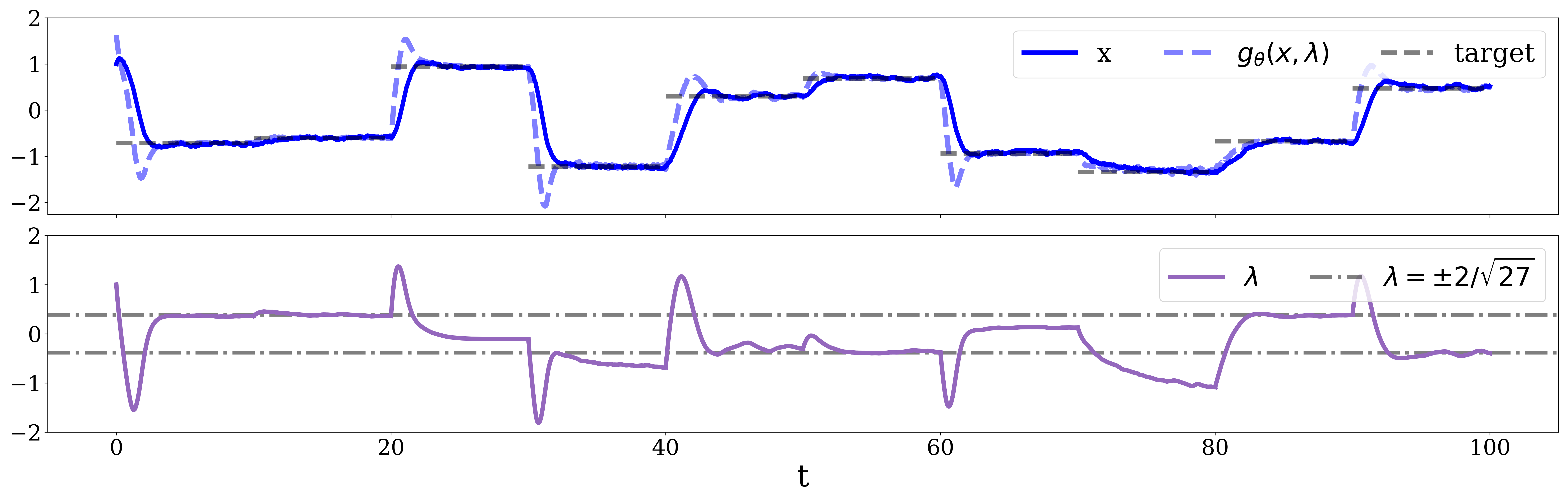}
    \caption{\textbf{Top:} Continuous feedback control for 10 randomized targets uniformly chosen from \([-1.5,1.5]\) where implicit equilibria function \(g_\theta(x,\lambda)\) included to show the behavior around tipping point. The states were perturbed by noise with scale \(\sigma = 0.03\). \textbf{Bottom:} Continuous control with control iteration and strength given by \(k=1\) and \(\eta=5\), as defined in \cref{eqn:cont-control-dyn}.
    Grey dashed horizontal lines denote tipping points of system. }
    \label{fig:sym-hyst-ref}
\end{figure}

\begin{figure}
    \centering
    \includegraphics[width=1\linewidth]{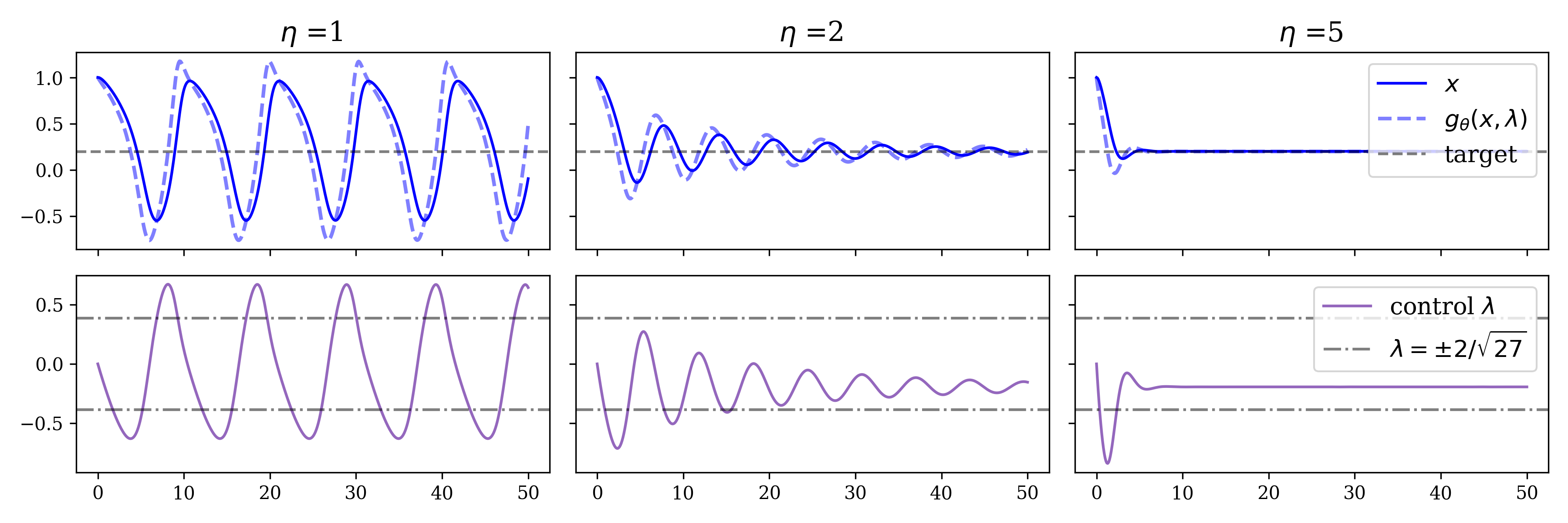}
    \caption{Comparison of effect of \(\eta\) for continuous control from \cref{eqn:cont-control-dyn} for unstable target, \(x = 0.2\), traversing tipping points, \(\lambda = \pm 2/ \sqrt{27}\), of symmetric hysteresis loop. The effect of control strength, \(\eta\), from an underdamped control to sufficiently strong control to mitigate the effects of tipping point (left to right).}
    \label{fig:eta-effects}
\end{figure}

\subsection{Budworm population} \label{sec:budworm}
We analyze the spruce budworm population model, a classic insect outbreak system originally established in \cite{ludwig1978qualitative}. For mathematical convenience, we adopt the dimensionless form presented in \cite{strogatz2014nonlinear}, which reduces the system to two parameters, \(r\) and \(\kappa\):

\begin{equation}
    \frac{\mathrm{d}x}{\mathrm{d}t} = r x\left(1-\frac{x}{\kappa}\right) - \frac{x^2}{1+x^2}.
\end{equation}
By fixing \(r=0.56\) and treating the carrying capacity, \(\kappa\), as the control parameter, the system exhibits tipping points at \(\kappa \approx 6.45\) and \(\kappa \approx 9.93\), with bistability arising for intermediate values.

The additional challenge of this system is that the large number of steady states along the upper branch often causes models to favor learning only a single branch. 

Similar to the example in \Cref{subsec:sym-hys}, this system exhibits hysteresis but presents the additional challenge of asymmetric steady-state branches surrounding the tipping points, as shown in \Cref{fig:budworm-bifur}. This asymmetry can bias the learning process to favor the larger branch while neglecting the smaller one, thereby failing to capture the full hysteretic dynamics. However, we demonstrate the ability to capture the full dynamics on limited transient data.

The system can be written in the decoupled form of \cref{eqn:fandg}:
\begin{equation}
    \begin{split}
        f(x) &= -\frac{x}{1+x^2}, \quad x > 0,\\
        g(x,\kappa) &= \frac{r}{\kappa}(1+x^2)(\kappa-x).
    \end{split}
\end{equation}
where we emphasize this decoupling is not unique and may not be the optimal form for learning a practical approximation.

For training, we simulate trajectories  up to \(t = 10\) from 51 uniformly sampled initial conditions \(x \in [0.1,10]\) and 51 control values \(\kappa \in [4.45, 11.99]\), slightly beyond the tipping point range. 

\subsubsection*{System identification} 
We use a gradient matching approach, namely \(\mathcal{L}_{\text{Grad}}\) from \cref{eqn:train-obj}, and perform a cross validation search over 10 folds for 200 epochs, batch size 50, and  an initial learning rate of \(\eta = 0.1\). The selected architecture was a fully connected 1-20-20-1 and 6-20-20-1 MLP, with bounds \((-4,-0.1)\) and \((-5,2)\) for \(f_\theta\) and \(g_\theta\), respectively. 
 The state features are constructed as in \cref{eqn:feat-x}, for the first five inputs to \(g_\theta\), with \(a=-1\) and \(b=1.5\) and the last input reserved for the control parameter \(\kappa\). The selected architecture achieved an average best validation loss of \(1.285 \times 10^{-2}\). The model was then trained on the full dataset for 500 epochs achieving a best loss of \(1.444 \times 10^{-6}\).

 The resulting zero level sets of \(r(x) = x - g_\theta(x,\kappa)\) are shown in \Cref{fig:budworm-zero-set}, and the identified bifurcation diagram is shown in \Cref{fig:budworm-bifur}.

\begin{figure}
    \centering
    \includegraphics[width=1\linewidth]{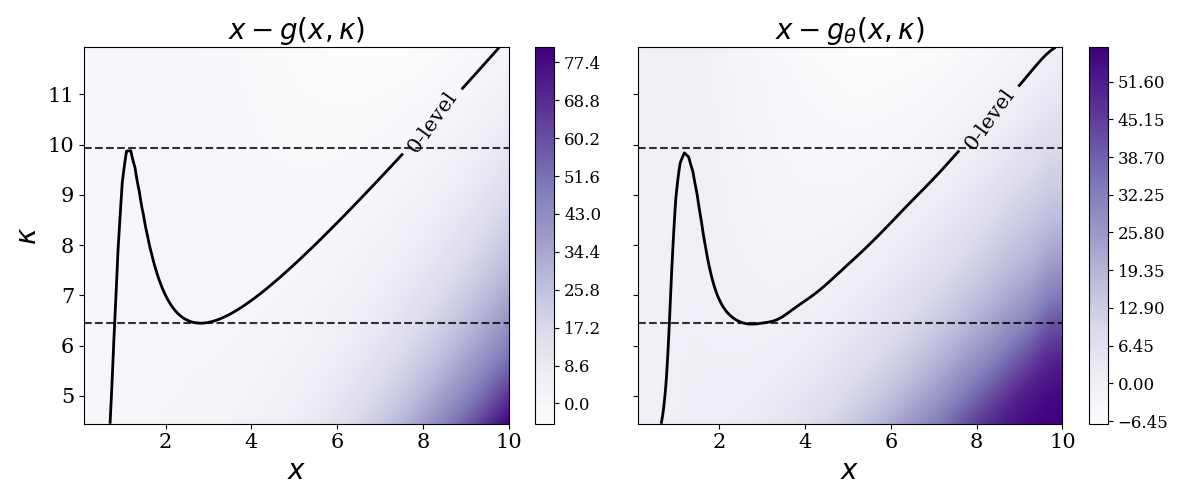}
    \caption{Comparison of the true and learned zero-level sets across \(x \in [0,10]\) and \(\kappa \in [4.44,11.93]\). \textbf{Left:} The true zero-level set (black solid line) of \(x-g(x,\kappa)\). \textbf{Right:} The zero-level set (black solid line) of the learned dynamics \(x-g_\theta(x,\kappa)\). Black dashed lines denote the system tipping points at \(\kappa = 6.446\) and \(9.934\).}
    \label{fig:budworm-zero-set}
\end{figure}

\begin{figure}
    \centering
    \includegraphics[width=0.7\linewidth]{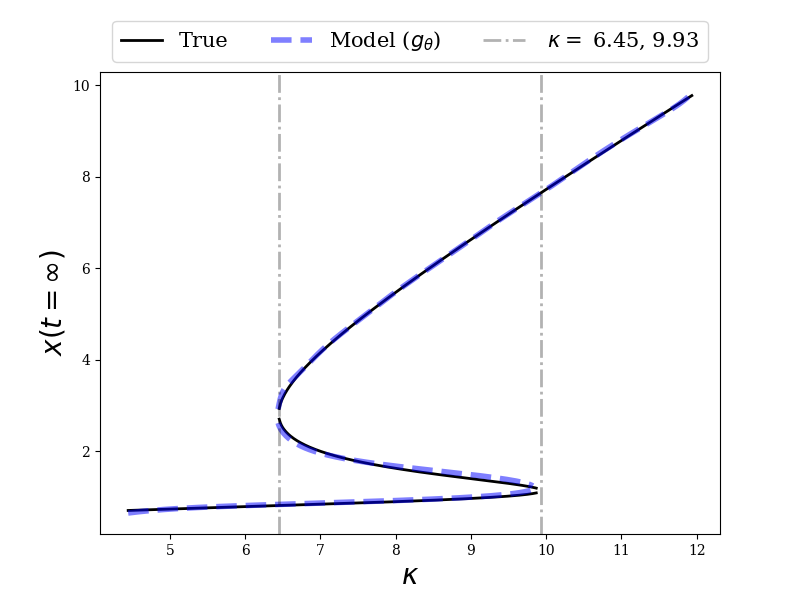}
    \caption{Learned bifurcation diagram from solving associated root finding problem \(r(x)=x- g_\theta(x,\kappa) \) for \(\kappa\in[4.44,11.93]\). Stability of each steady state is determined by evaluating \(\frac{\partial}{\partial x} [f(x)(x-g(x,\kappa))]_{x=x^*}\). Stable states are denoted in blue, unstable in red. Dashed lines denote learned model, black solid lines denote true, numerically simulated results. Gray dashed vertical lines denote system tipping points.}
    \label{fig:budworm-bifur}
\end{figure}

\subsubsection*{Feedback control} 
We validate the feedback control performance by selecting \(10\) randomized targets uniformly from the range \([0.1,10]\) across \(100\) trials. The average nRMSE is \(2.209 \cdot 10^{-2}\) with standard deviation \(3.302 \times 10^{-2}\). \(90\%\) of targets are reached with \(5\%\) relative error of the system magnitude, while \(77\%\) reach the target within \(2\%\). An example feedback control experiment is shown in \Cref{fig:budworm-ref-track}.

\begin{figure}
    \centering
    \includegraphics[width=1\linewidth]{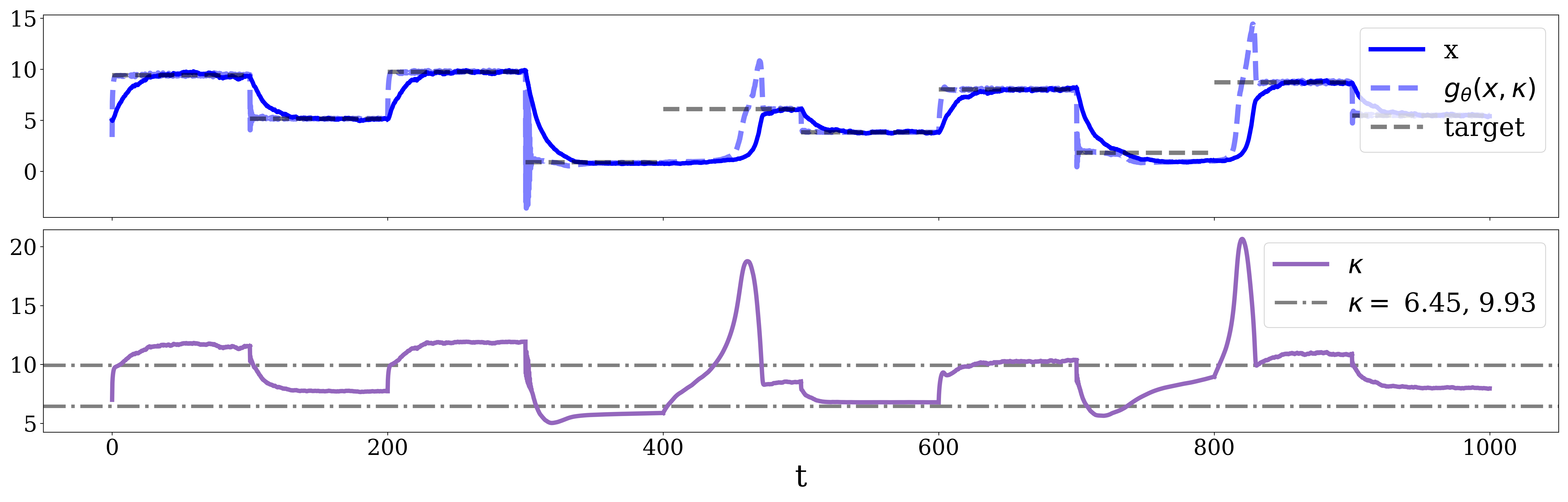}
    \caption{\textbf{Top:} Continuous feedback control for the Budworm ODE given 10 randomized targets perturbed by noise with scale \(\sigma= 0.02\). \textbf{Bottom:} Continuous control with control iteration and strength are given by \(k=1\) and \(\eta=20\), respectively, as defined in \cref{eqn:cont-control-dyn}. Grey dashed horizontal lines denote tipping points of system, \(\kappa = 6.45, 9.93\).}
    \label{fig:budworm-ref-track}
\end{figure}

\subsection{Genetic toggle switch}

We consider the synthetic, bistable gene-regulatory network originally proposed in the work in \cite{gardner_construction_2000} to model \textit{Escherichia coli}. This system is comprised of two promoters that are inhibited by two repressors transcribed by the opposing promoter; an overview can be seen in \Cref{fig:toggle_switch}. The dimensionless system is given by 
\begin{equation}
    \begin{split}
        \frac{\mathrm{d}x_1}{\mathrm{d}t} &= -x_1 + \frac{\alpha_1}{1+x_2^\beta}\\
        \frac{\mathrm{d} x_2}{\mathrm{d}t} &= -x_2 + \frac{\alpha_2}{1+x_1^\gamma}
    \end{split}
    \label{eqn:genetic-toggle}
\end{equation}
where \(x_1\) and \(x_2\) are the concentrations of two active repressor/promoters, \(\alpha_1\) and \(\alpha_2\) are the effective synthesis rates of each repressor, and \(\beta, \gamma\) are the cooperativity (Hill Coefficients) of each repressor, respectively. The model is flexible in modeling various biological toggle switch systems and gives rise to a non-trivial set of bifurcations and hysteresis \cite{ANGELI2004185,li_multistability_2013}. This can be seen explicitly in \Cref{fig:toggle-switch-results} where, as an example, confining control of the system to changing only one parameter, e.g. \(\alpha_1\), and fixing all other parameters, \(\alpha_2, \beta, \gamma\), illustrates coupled hysteretic bifurcation diagrams for the steady state of both \(x_1\) and \(x_2\). Increasing the control dimension immediately exacerbates the challenging nature of this problem from a system id and control perspective.

For model training, variable transient times (time to steady state) were determined via reverse gradient accumulation on trajectories simulated up to \(t=100\). The training set was generated using a full combinatorial design of \(81\) initial conditions (\(9\times 9\) grid over \([0,6]^2\)) and all permutations of the control parameters:
\begin{equation*}
    \alpha_1, \alpha_2, \beta, \gamma \in \{0.1, 1.25, 2.5 , 3.75, 5\}.
\end{equation*}

\subsubsection*{System identification}
Using the gradient matching approach, \(\mathcal{L}_{\text{Grad}}\) in \cref{eqn:train-obj}, we performed a 10-fold cross validation across multiple architectures and random seeds. All models were initialized with the same training setup: training for 500 epochs, batch size of 200, and an initial learning rate of \(\eta = 0.01\).

We observed a disconnect between average validation scores and full-dataset performance. While shallower networks achieved higher average cross validation scores, they consistently plateaued and failed to capture the dynamics when trained on the full dataset. Conversely, deeper MLPs generally outperformed all other models across the majority of individual folds, but their average validation scores were often drastically penalized by typically a single fold that failed to learn. We attribute this single-fold failure to the system's high number of bifurcations; omitting a single fold of training data can easily exclude critical topological features, such as an entire hysteresis lobe, rendering the network unable to generalize to that validation set.

Recognizing this artifact of the data partitioning, and consistent with preliminary investigations across all our experiments that generally favored deeper models, we opted for the deeper MLP architectures: 2-20-20-20-2 for \(f_\theta\) and 6-20-20-20-2 for \(g_\theta\). When trained on the complete dataset, this configuration far outperformed all other candidates, achieving a best loss of \(7.570\times 10^{-4}\).

From an analytical perspective we note the system given by \cref{eqn:genetic-toggle} can be written cleanly in the form  

\begin{equation*}
    f(x_1, x_2) = (-1,-1)^\top \quad \text{and} \quad g(x_1, x_2, \alpha_1, \alpha_2, \beta,\gamma) = \left(\frac{\alpha_1}{1+x_2^\beta}, \frac{\alpha_2}{1+x_1^\gamma}\right)^\top
\end{equation*}

where \(g\) is sigmoid-like and univariate across both dimensions. When learning a candidate \(f_\theta\) and \(g_\theta\) solely from trajectory observations with no additional information a priori, other than the proposed learning structure, it would be ideal to learn a  constant \(f=-1\), and a \(g\) that is dominantly univariate. In fact, this is recovered nearly identically; shown in \Cref{fig:f-and-g-toggle-switch}. The learned structure correctly captures the dynamics and multistable regimes of the system as shown in \Cref{fig:toggle-switch-results}, where we can identify hysteretic bifurcations across various control values using only \(g_\theta\), as well as simulated non-trivial trajectories as seen in \Cref{fig:toggle-traj}.

\subsubsection*{Feedback control}
Employing \(g_\theta\) as a feedback controller in this context is notably more challenging due to the increased number of control parameters and the coupled hysteretic dynamics. We utilize a constrained control policy, \cref{eqn:constrained-control}, to enforce positivity and restrict the cooperativity coefficients to be strictly greater than 1. This constraint serves two purposes:

\begin{enumerate}
    \item It encourages the system to operate within the more dynamically interesting multistable regions.
    \item It demonstrates the ability to steer dynamics toward targets generated by control values distinct from those available to the controller (i.e., targets generated by parameters the controller is forbidden from using). This mimics real-world experimental constraints and highlights the robustness of the approach.
\end{enumerate}

Specifically, for controls $\vect{u} = (\alpha_1, \alpha_2, \beta, \gamma)$, we define the element-wise control constraint $\phi(u_i)$ as:
\begin{equation}
    \phi(u_i) = 
    \begin{cases} 
        H(u_i - \delta_1, \ell), & u_i \in \{\alpha_1, \alpha_2\} \\
        H(u_i - \delta_2, \ell), &  u_i \in \{\beta, \gamma\}
    \end{cases}
    \label{eqn:toggle-control-constraint}
\end{equation}
where $\delta_1 = 0.1$, $\delta_2 = 1.1$, and the decay parameter $\ell = 200$ is applied to our feedback policy as in \cref{eqn:constrained-control}.

We visualize the loss landscape of \(g_\theta\) in \Cref{fig:toggle-loss-landscape}, demonstrating the model correctly captures both the location and number of steady states. We also present constrained control experiments over 100 trials, steering the restricted dynamics toward 10 randomized targets. The targets were generated using initial conditions and control values drawn uniformly from \([0,6]^2\) and \([0,5]^4\), respectively.

Despite the complex coupled hysteretic dynamics and the constrained control policy, the learned feedback controller successfully steered the system to the designated targets with high accuracy. Specifically, the average nRMSE, over all 100 trials, was \(5.746 \times 10^{-2}\) (with a standard deviation of \(4.770 \times 10^{-2}\)) for \(x_1\), and \(5.275 \times 10^{-2}\) (with a standard deviation of \(4.379 \times 10^{-2}\)) for \(x_2\). Demonstrating the robustness of the controller, 99\% of the achieved steady states fell within a 5\% error margin relative to the system magnitude, with 85\% and 91\% of the trials for \(x_1\) and \(x_2\), respectively, achieving sub-2\% error. An example feedback control experiment is shown in \Cref{fig:feedback-toggle}.

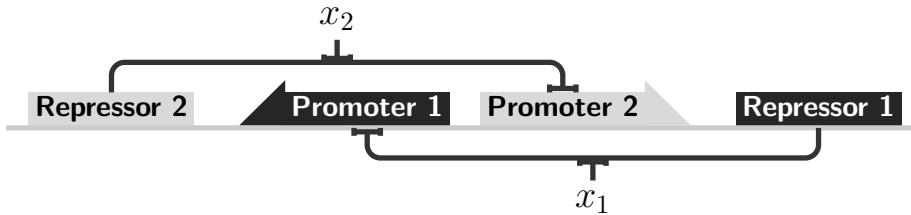
\begin{figure}[t]
    \centering
    \begin{tikzpicture}[
        node distance=0pt, 
        font=\sffamily\bfseries,
        gene/.style={
            draw=none,
            rectangle,
            minimum height=0.35cm,
            text centered,
            outer sep=0pt,
            inner sep=2pt,
            text height=1.5ex,
            text depth=0.5ex
        },
        repressor/.style={
            gene,
            minimum width=2.2cm,
            text=black
        },
        promoter/.style={
            gene,
            minimum width=2.2cm
        },
        inhibitor line/.style={
            draw=black!80,
            line width=2pt,
            rounded corners=5pt,
            -{Bar[width=12pt, length=3pt]}
        },
        inducer label/.style={
            font=\Large,
            text=black
        }
    ]

    \node[repressor, fill=gray!30] (rep2) {Repressor 2};
    \node[promoter, fill=black!85, text=white, right=1.2cm of rep2] (prom1) {Promoter 1};
    \node[promoter, fill=gray!30, right=0.4cm of prom1] (prom2) {Promoter 2};
    \node[repressor, fill=black!85, text=white, right=1.2cm of prom2] (rep1) {Repressor 1};

    \draw[line width=2pt, gray!40] 
        ($(rep2.south west)+(-0.3,0.0)$) -- ($(rep1.south east)+(0.3,0.0)$);

    \fill[black!85] 
        ($(prom1.west)+(0.01, 0.4)$) --      
        ($(prom1.west)+(0.01, -0.201)$) --     
        ($(prom1.west)+(-0.6, -0.201)$) -- 
        cycle;

    \fill[gray!30] 
        ($(prom2.east)+(-0.01, 0.4)$) --      
        ($(prom2.east)+(-0.01, -0.201)$) --     
        ($(prom2.east)+(0.60, -0.201)$) --  
        cycle;

    \def\voffset{0.4cm}
    \def\labeloffset{0.3cm}

    \draw[inhibitor line] (rep1.south) 
        -- ++(0,-\voffset) coordinate (bot_corner1) 
        -| (prom1.south);

    \draw[inhibitor line] (rep2.north) 
        -- ++(0,\voffset) coordinate (top_corner1) 
        -| (prom2.north);

    \path (bot_corner1) -- (bot_corner1 -| prom1.south) coordinate[midway] (mid_bot);
    \node[inducer label, below=\labeloffset of mid_bot] (ind1) {$x_1$};
    \draw[inhibitor line] (ind1.north) -- (mid_bot);

    \path (top_corner1) -- (top_corner1 -| prom2.north) coordinate[midway] (mid_top);
    \node[inducer label, above=\labeloffset of mid_top] (ind2) {$x_2$};
    \draw[inhibitor line] (ind2.south) -- (mid_top);

    \end{tikzpicture}

    \caption{Toggle switch design. Repressor 1 inhibits transcription from Promoter 1 and is induced by \(x_1\). Repressor 2 inhibits transcription from Promoter 2 and is induced by \(x_2\).}
    \label{fig:toggle_switch}
\end{figure}

\begin{figure}
    \centering
    \includegraphics[width=1\linewidth]{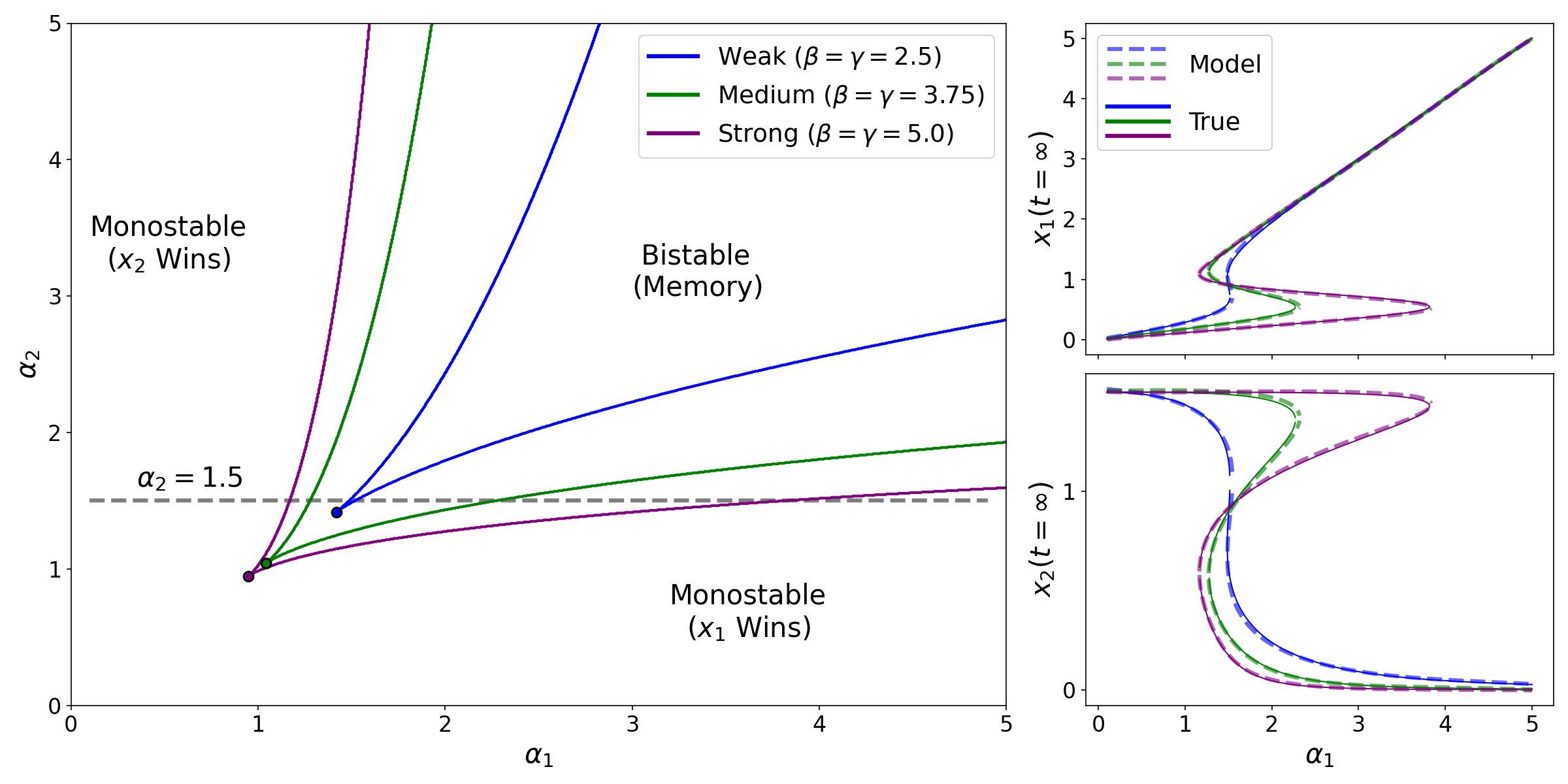}
    \caption{Comparison of true and modeled bistability landscapes across varying cooperativity levels. \textbf{Left:} Overlay of bistability regions (``Arnold tongues'') in the $(\alpha_1, \alpha_2)$ parameter space for symmetric cooperativity cases ($\beta=\gamma$). The wedge shapes delineate the bistable (memory) regime from the monostable regimes. The dashed horizontal line at $\alpha_2=1.5$ indicates the cross-section used for the bifurcation analysis on the right. \textbf{Right:} Steady-state response of $x_1$ (top) and $x_2$ (bottom) as a function of $\alpha_1$, with $\alpha_2$ fixed at $1.5$. Solid lines represent the true system dynamics, while dashed lines represent the learned model predictions using \(g_\theta\). The colors correspond to the cooperativity levels defined in the left panel.}
    \label{fig:toggle-switch-results}
\end{figure}

\begin{figure}
    \centering
    \includegraphics[width=1\linewidth]{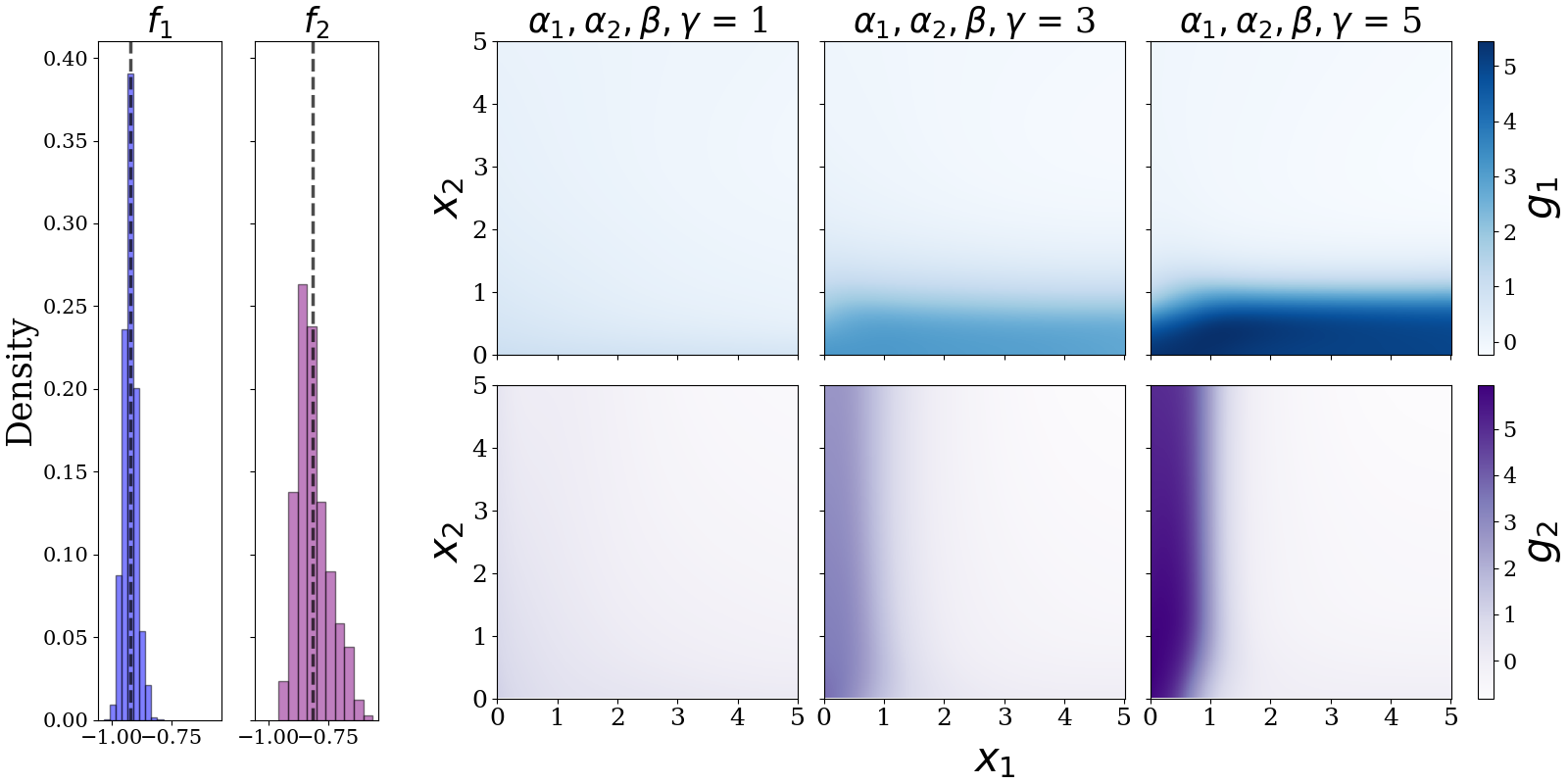}
    \caption{\textbf{Left:} Histograms of the learned function \(f_\theta\) evaluated over a dense grid where \(x_1, x_2 \in [0,5]\). Black vertical dashed lines indicate mean values of \(-0.9224\) and \(-0.8152\), with standard deviations \(2.720\cdot 10^{-2}\) and \(6.807 \cdot 10^{-2}\), respectively. \textbf{Right:} Output of the learned control function \(g_\theta\) on the same domain for increasing control parameters \(\alpha_1 = \alpha_2 = \beta = \gamma \in \{1, 3, 5\}\). The top and bottom rows display the first and second output dimensions of \(g_\theta\), respectively. Notably, the learned dynamics exhibit univariate behavior.}
    \label{fig:f-and-g-toggle-switch}
\end{figure}

\begin{figure}
    \centering
    \includegraphics[width=0.9\linewidth]{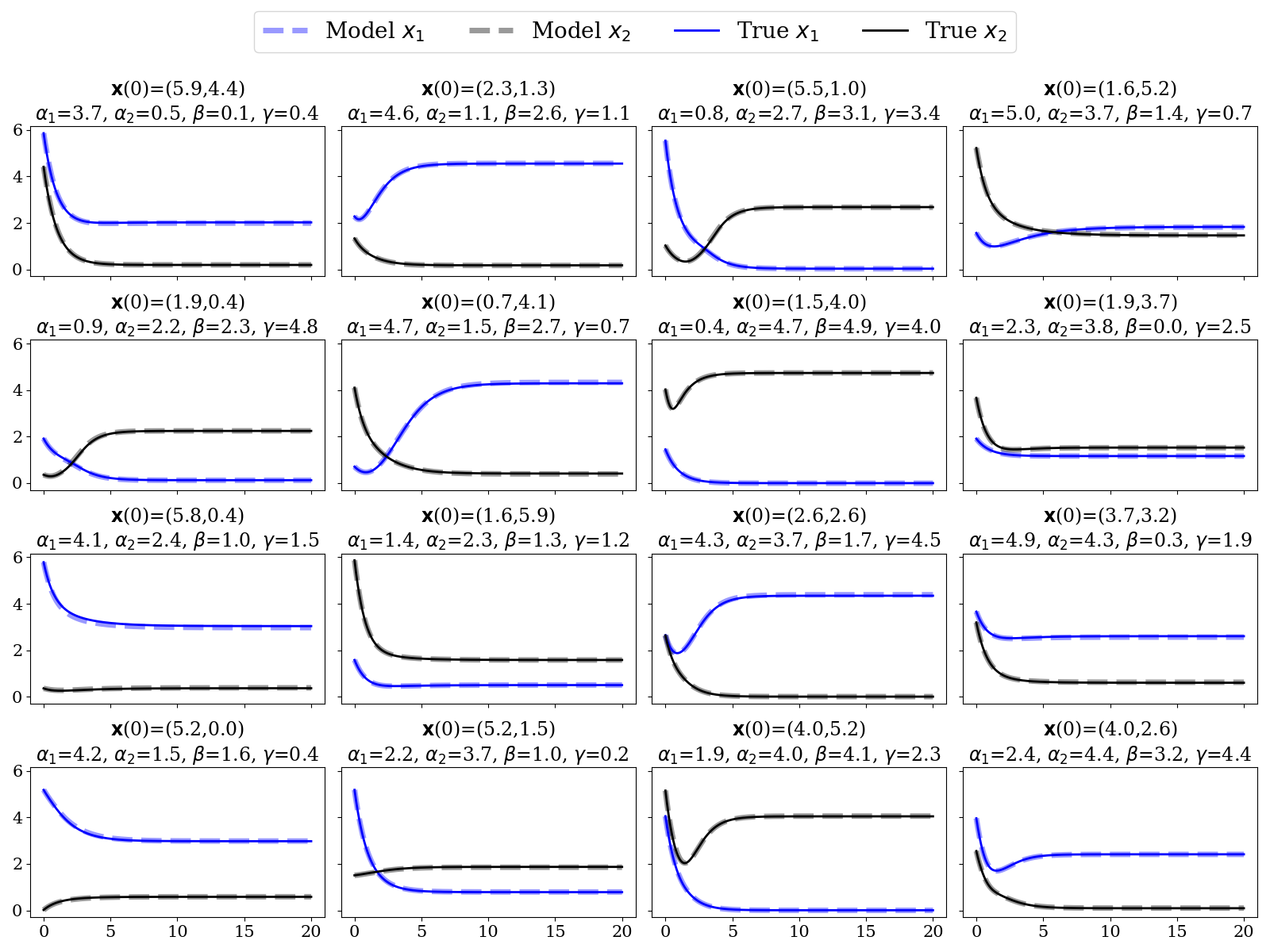}
    \caption{Comparison of true versus learned simulated trajectories: using full \(f_\theta(\vect{x})(x-g_\theta(\vect{x},\vect{u}))\). Initial conditions are sampled uniformly from \(\vect{x}(0) \in [0,6]^2\), while control parameters \((\alpha_1, \alpha_2, \beta, \gamma)\) are drawn uniformly from \([0,5]^4\).}
    \label{fig:toggle-traj}
\end{figure}

\begin{figure}
    \centering
    \includegraphics[width=1\linewidth]{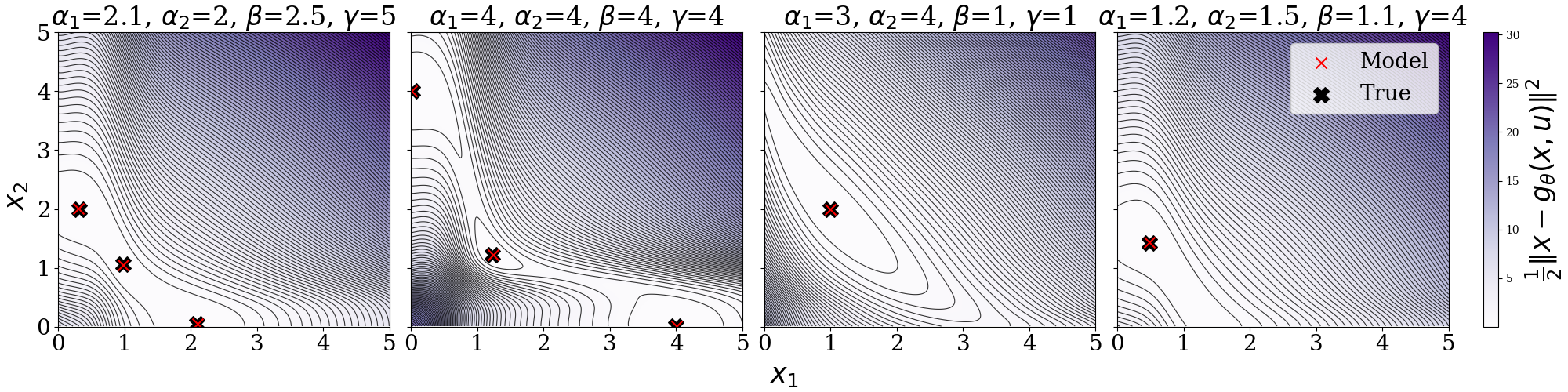}
    \caption{Feedback control loss landscape for various control configurations. Black crosses denote true minima found and red crosses denote minima found using \(g_\theta\) via solving corresponding objective: \(\frac{1}{2}\|\vect{x} - g_\theta(\vect{x},\vect{u})\|^2\). A coarse grid search and subsequent root finding scheme (\texttt{scipy.optimize.root}) was used to determine solutions for both the true and learned model.}
    \label{fig:toggle-loss-landscape}
\end{figure}

\begin{figure}
    \centering
    \includegraphics[width=1\linewidth]{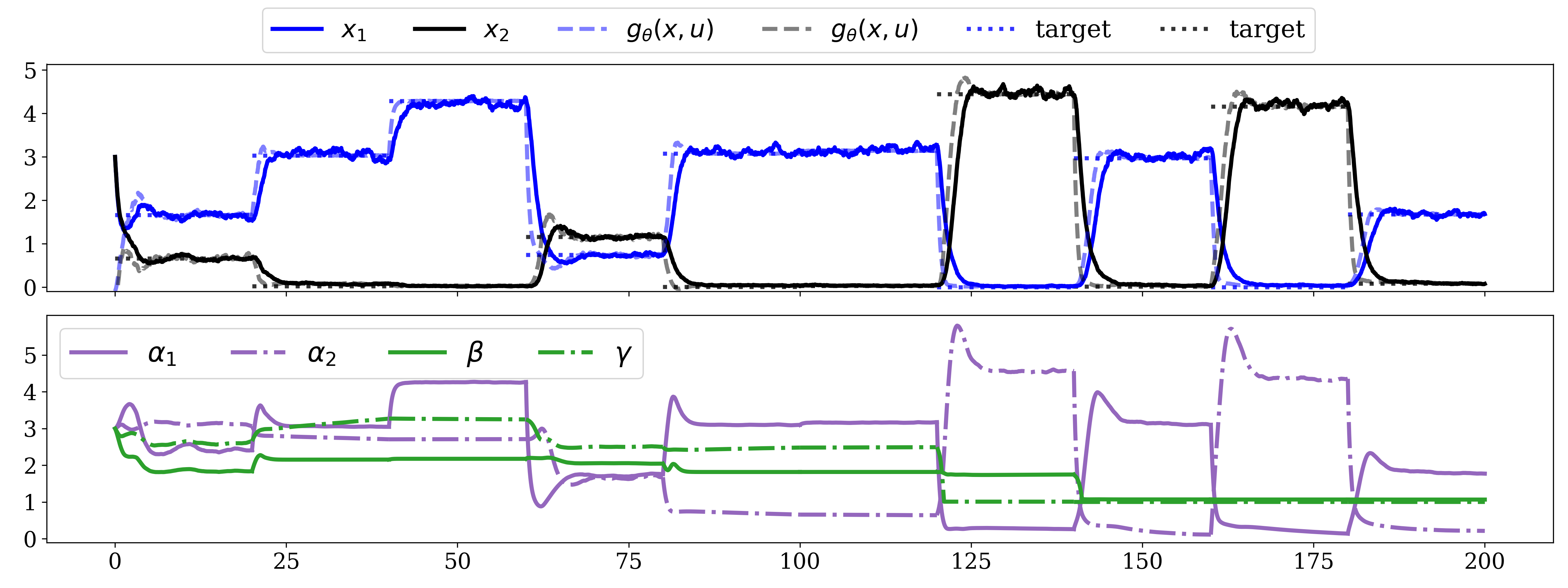}
    \caption{\textbf{Top:} Control restricted feedback control experiment for genetic toggle switch system \cref{eqn:genetic-toggle} for 10 randomized targets every \(t=20\)  perturbed by noise with scale \(\sigma = 0.05\). \textbf{Bottom:} Continuous control with control iteration and strength are given by \(k=1\) and \(\eta = 1\) as defined in \cref{eqn:constrained-control} with constraints given by \cref{eqn:toggle-control-constraint}.}
    \label{fig:feedback-toggle}
\end{figure}

\section{Discussion and future work}
While the proposed stability-constrained framework demonstrates robust performance on learning multistable systems and feedback control over several bifurcations, there remain practical challenges and theoretical extensions. We highlight key limitations of the current approach and outline promising directions for ongoing research. 

\subsection*{Density of state-space observations and trajectory frequency}
The effects of data scarcity and the density of control observations require further investigation. While an exhaustive study was outside the scope of this work, our preliminary results suggest that a dense set of trajectory observations covering the control space is necessary for the current implementation. We leave a detailed analysis of data efficiency to future work.

\subsection*{Noisy dynamics and future architectures}
Although not explicitly investigated here, standard denoising techniques could be applied prior to model learning, integrated into the cross validation and model selection pipeline.

Furthermore, in scenarios involving data scarcity or significant noise---in either state trajectories or control inputs---we are particularly interested in an \emph{All-at-once} approach, similar to the method proposed in \cite{hsu2026jointoptimizationapproachidentifying}, where system identification and state estimation are solved jointly. Adapting this framework by substituting the system identification component with our proposed architecture (augmented with a state interpolator) remains a promising direction for future research.

Finally, we are interested in exploring the learning and rigorous comparison of optimal control policies. While we have demonstrated the potential effectiveness of our proposed strategy for learning control policies, we defer a more rigorous theoretical and empirical investigation to future work.

\section{Conclusion}

We introduced a stability-constrained NODE framework for learning dynamics of nonlinear systems that are asymptotically stable from trajectory data. By enforcing the structured form
\[
F(\vect{x},\vect{u}) = f(\vect{x})\,(\vect{x} - g(\vect{x},\vect{u})), \qquad f(\vect{x}) < 0,
\]
the model guarantees trajectory stability, offers additional interpretability, and yields a tractable parameterization for multi-attractor and hysteretic dynamics from short time-horizon data.

This structure enables accurate reconstruction of equilibrium landscapes, bifurcations, and hysteresis loops, even in regimes where long trajectories or rich excitation are unavailable. The learned equilibrium map \(g_\theta\) further provides an efficient and robust control mechanism: gradient-based feedback control policies directly optimize steady-state objectives and reliably traverse tipping points. Across mixing-tank and hysteretic benchmarks, the method achieves accurate dynamics learning and consistent closed-loop control, including steering to both stable and unstable equilibria under noise perturbations.

\section*{Acknowledgments}
This research was supported by the AT SCALE initiative via the Laboratory Directed Research and Development (LDRD) investments at Pacific Northwest National Laboratory (PNNL). PNNL is a national laboratory operated for the U.S. Department of Energy (DOE) by Battelle Memorial Institute under Contract No. DE-AC05-76RL0-1830.

\bibliographystyle{unsrtnat}   
\bibliography{refs}

@inproceedings{chen2018neural,
  author = {Chen, Tian Qi and Rubanova, Yulia and Bettencourt, Jesse and Duvenaud, David},
  booktitle = {NeurIPS},
  editor = {Bengio, Samy and Wallach, Hanna M. and Larochelle, Hugo and Grauman, Kristen and Cesa-Bianchi, Nicolò and Garnett, Roman},
  ee = {http://papers.nips.cc/paper/7892-neural-ordinary-differential-equations},
  pages = {6572-6583},
  title = {Neural Ordinary Differential Equations.},
  url = {http://dblp.uni-trier.de/db/conf/nips/nips2018.html#ChenRBD18},
  year = 2018
}

@misc{okamoto2025learningsimplestneuralode,
      title={Learning the Simplest Neural ODE}, 
      author={Yuji Okamoto and Tomoya Takeuchi and Yusuke Sakemi},
      year={2025},
      eprint={2505.02019},
      archivePrefix={arXiv},
      primaryClass={stat.ML},
      url={https://arxiv.org/abs/2505.02019}, 
}

@article{aastrom1971system,
  title={System identification—a survey},
  author={{\AA}str{\"o}m, Karl Johan and Eykhoff, Peter},
  journal={Automatica},
  volume={7},
  number={2},
  pages={123--162},
  year={1971},
  publisher={Elsevier}
}

@article{ljung2010perspectives,
  title={Perspectives on system identification},
  author={Ljung, Lennart},
  journal={Annual Reviews in Control},
  volume={34},
  number={1},
  pages={1--12},
  year={2010},
  publisher={Elsevier}
}

@book{keesman2011system,
  title={System identification: an introduction},
  author={Keesman, Karel J},
  year={2011},
  publisher={Springer Science \& Business Media}
}

@article{bongard2007automated,
author = {Bongard, Josh and Lipson, Hod},
year = {2007},
month = {07},
pages = {9943-8},
title = {Automated reverse engineering of nonlinear dynamical systems},
volume = {104},
journal = {Proceedings of the National Academy of Sciences of the United States of America}
}

@article{schmidt2009distilling,
author = {Schmidt, Michael and Lipson, Hod},
year = {2009},
month = {05},
pages = {81-5},
title = {Distilling Free-Form Natural Laws from Experimental Data},
volume = {324},
journal = {Science (New York, N.Y.)}
}

@article{schaeffer2017learning,
  title={Learning partial differential equations via data discovery and sparse optimization},
  author={Schaeffer, Hayden},
  journal={Proceedings of the Royal Society A: Mathematical, Physical and Engineering Sciences},
  volume={473},
  number={2197},
  pages={20160446},
  year={2017},
  publisher={The Royal Society Publishing}
}

@article{rudy2017data,
  title={Data-driven discovery of partial differential equations},
  author={Rudy, Samuel H and Brunton, Steven L and Proctor, Joshua L and Kutz, J Nathan},
  journal={Science Advances},
  volume={3},
  number={4},
  pages={e1602614},
  year={2017},
  publisher={American Association for the Advancement of Science},
  doi={10.1126/sciadv.1602614},
  url={https://www.science.org/doi/abs/10.1126/sciadv.1602614}
}

@article{wang2011,
  title = {Predicting Catastrophes in Nonlinear Dynamical Systems by Compressive Sensing},
  author = {Wang, Wen-Xu and Yang, Rui and Lai, Ying-Cheng and Kovanis, Vassilios and Grebogi, Celso},
  journal = {Phys. Rev. Lett.},
  volume = {106},
  issue = {15},
  pages = {154101},
  numpages = {4},
  year = {2011},
  month = {4},
  publisher = {American Physical Society},
  doi = {10.1103/PhysRevLett.106.154101},
  url = {https://link.aps.org/doi/10.1103/PhysRevLett.106.154101}
}

@article{Brunton2016,
    author = {Brunton, Steven L and Proctor, Joshua L and {Nathan Kutz}, J},
    doi = {10.1073/pnas.1517384113},
    journal = {Proceedings of the National Academy of Sciences},
    number = {15},
    title = {{Discovering governing equations from data by sparse identification of nonlinear dynamical systems}},
    volume = {113},
    year = {2016}
}

@article{donoho2006compressed,
  title        = {Compressed sensing},
  author       = {Donoho, David L.},
  journal      = {IEEE Transactions on Information Theory},
  volume       = {52},
  number       = {4},
  pages        = {1289--1306},
  year         = {2006},
  doi          = {10.1109/TIT.2006.871582}
}

@article{candes2006robust,
  title        = {Robust uncertainty principles: exact signal reconstruction from highly incomplete frequency information},
  author       = {Cand{\`e}s, Emmanuel J. and Romberg, Justin and Tao, Terence},
  journal      = {IEEE Transactions on Information Theory},
  volume       = {52},
  number       = {2},
  pages        = {489--509},
  year         = {2006},
  doi          = {10.1109/TIT.2005.862083}
}

@article{candes2008intro,
  title        = {An introduction to compressive sampling},
  author       = {Cand{\`e}s, Emmanuel J. and Wakin, Michael B.},
  journal      = {IEEE Signal Processing Magazine},
  volume       = {25},
  number       = {2},
  pages        = {21--30},
  year         = {2008},
  doi          = {10.1109/MSP.2007.914731}
}

@article{PILLONETTO2025111907,
	author = {Gianluigi Pillonetto and Aleksandr Aravkin and Daniel Gedon and Lennart Ljung and Ant{\^o}nio H. Ribeiro and Thomas B. Sch{\"o}n},
	doi = {https://doi.org/10.1016/j.automatica.2024.111907},
	issn = {0005-1098},
	journal = {Automatica},
	pages = {111907},
	title = {Deep networks for system identification: A survey},
	url = {https://www.sciencedirect.com/science/article/pii/S0005109824004011},
	volume = {171},
	year = {2025}}

@misc{shumaylov2025discoverabledatadiscoveryrequires,
      title={When is a System Discoverable from Data? Discovery Requires Chaos}, 
      author={Zakhar Shumaylov and Peter Zaika and Philipp Scholl and Gitta Kutyniok and Lior Horesh and Carola-Bibiane Schönlieb},
      year={2025},
      eprint={2511.08860},
      archivePrefix={arXiv},
      primaryClass={math.DS},
      url={https://arxiv.org/abs/2511.08860}, 
}

@article{LJUNG20201175,
    author = {Lennart Ljung and Carl Andersson and Koen Tiels and Thomas B. Sch{\"o}n},
    title = {Deep Learning and System Identification},
    journal = {IFAC-PapersOnLine},
    volume = {53},
    number = {2},
    pages = {1175--1181},
    year = {2020},
    note = {21st IFAC World Congress},
    issn = {2405-8963},
    doi = {10.1016/j.ifacol.2020.12.1329}
}

@misc{koch2025contactdatadrivenfrictionstirprocess,
      title={First Contact: Data-driven Friction-Stir Process Control}, 
      author={James Koch and Ethan King and WoongJo Choi and Megan Ebers and David Garcia and Ken Ross and Keerti Kappagantula},
      year={2025},
      eprint={2507.03177},
      archivePrefix={arXiv},
      primaryClass={eess.SY},
      url={https://arxiv.org/abs/2507.03177}, 
}

@article{Koch2025Neural,
  author    = {James Koch and WoongJo Choi and Ethan King and David Garcia and Hrishikesh Das and Tianhao Wang and Ken Ross and Keerti Kappagantula},
  title     = {Neural lumped parameter differential equations with application in friction-stir processing},
  journal   = {Journal of Intelligent Manufacturing},
  year      = {2025},
  volume    = {36},
  number    = {2},
  pages     = {1111--1121},
  doi       = {10.1007/s10845-023-02271-5},
  url       = {https://doi.org/10.1007/s10845-023-02271-5}
}

@inproceedings{Kolter2019LearningSD,
  title={Learning Stable Deep Dynamics Models},
  author={J. Zico Kolter and Gaurav Manek},
  booktitle={Advances in Neural Information Processing Systems},
  volume={32},
  pages={11128--11136},
  year={2019},
  url={https://proceedings.neurips.cc/paper/2019/file/0a4bbceda17a6253386bc9eb45240e25-Paper.pdf}
}

@inproceedings{wu2023neural,
  title={Neural Lyapunov Control for Discrete-Time Systems},
  author={Wu, Jialin and Clark, Andrew and Kantaros, Yiannis and Vorobeychik, Yevgeniy},
  booktitle={Advances in Neural Information Processing Systems},
  volume={36},
  pages={66612--66623},
  year={2023},
  url={https://arxiv.org/abs/2305.06547}
}

@inproceedings{kang2021stableneuralodelyapunovstable,
author = {Kang, Qiyu and Song, Yang and Ding, Qinxu and Tay, Wee Peng},
title = {Stable neural ODE with lyapunov-stable equilibrium points for defending against adversarial attacks},
year = {2021},
isbn = {9781713845393},
publisher = {Curran Associates Inc.},
address = {Red Hook, NY, USA},
booktitle = {Proceedings of the 35th International Conference on Neural Information Processing Systems},
articleno = {1144},
numpages = {13},
series = {NIPS '21}
}

@inproceedings{okamoto2024learningdeepdissipativedynamics,
author = {Okamoto, Yuji and Kojima, Ryosuke},
title = {Learning deep dissipative dynamics},
year = {2025},
isbn = {978-1-57735-897-8},
publisher = {AAAI Press},
url = {https://doi.org/10.1609/aaai.v39i18.34175},
doi = {10.1609/aaai.v39i18.34175},
booktitle = {Proceedings of the Thirty-Ninth AAAI Conference on Artificial Intelligence and Thirty-Seventh Conference on Innovative Applications of Artificial Intelligence and Fifteenth Symposium on Educational Advances in Artificial Intelligence},
articleno = {2202},
numpages = {9},
series = {AAAI'25/IAAI'25/EAAI'25}
}

@INPROCEEDINGS{sochopoulos2024learningdeepdynamicalsystems,
  author={Sochopoulos, Andreas and Gienger, Michael and Vijayakumar, Sethu},
  booktitle={2024 IEEE/RSJ International Conference on Intelligent Robots and Systems (IROS)}, 
  title={Learning Deep Dynamical Systems using Stable Neural ODEs}, 
  year={2024},
  volume={},
  number={},
  pages={11163-11170},
  doi={10.1109/IROS58592.2024.10801826}}

@inproceedings{yan2020robustness,
    title={On Robustness of Neural Ordinary Differential Equations},
    author={Hanshu Yan and Jiawei Du and Vincent Y. F. Tan and Jiashi Feng},
    booktitle={International Conference on Learning Representations (ICLR)},
    year={2020},
    url={https://openreview.net/forum?id=B1e9Y2NYvS}
}

@inproceedings{bai2019deepequilibriummodels,
	author = {Bai, Shaojie and Kolter, J. Zico and Koltun, Vladlen},
	booktitle = {Advances in Neural Information Processing Systems},
	editor = {H. Wallach and H. Larochelle and A. Beygelzimer and F. d\textquotesingle Alch\'{e}-Buc and E. Fox and R. Garnett},
	publisher = {Curran Associates, Inc.},
	title = {Deep Equilibrium Models},
	url = {https://proceedings.neurips.cc/paper_files/paper/2019/file/01386bd6d8e091c2ab4c7c7de644d37b-Paper.pdf},
	volume = {32},
	year = {2019}}

@misc{DOE_PumpedStorageHydropower,
  author       = {U.S. Department of Energy},
  title        = {Pumped Storage Hydropower},
  howpublished = {\url{https://www.energy.gov/eere/water/pumped-storage-hydropower}},
  year         = {2025},
  note         = {Accessed: 2025-11-18},
  institution  = {U.S. Department of Energy, Office of Energy Efficiency & Renewable Energy}
}

@INPROCEEDINGS{koch2024atmospheric,

  author={Koch, James and Forland, Brenda and Bernacki, Bruce and Doster, Timothy and Emerson, Tegan},

  booktitle={IGARSS 2024 - 2024 IEEE International Geoscience and Remote Sensing Symposium}, 

  title={Data-Driven Invertible Neural Surrogates of Atmospheric Transmission}, 

  year={2024},

  volume={},

  number={},

  pages={6943-6947},

  doi={10.1109/IGARSS53475.2024.10642124}}

@misc{koch2024gnde,
  title        = {Graph Neural Differential Equations for Coarse-Grained Socioeconomic Dynamics},
  author       = {Koch, James and Roy Chowdhury, Pranab K. and Wan, Heng and Bhaduri, Parin and Yoon, Jim and Srikrishnan, Vivek and Daniel, W. Brent},
  year         = {2024},
  eprint       = {2407.18108},
  archivePrefix= {arXiv},
  primaryClass = {cs.LG},
  url          = {https://arxiv.org/abs/2407.18108}
}

@inproceedings{Lienen_torchode_A_Parallel_2022,
author = {Lienen, Marten and Günnemann, Stephan},
booktitle = {The Symbiosis of Deep Learning and Differential Equations II, NeurIPS},
title = {{torchode: A Parallel ODE Solver for PyTorch}},
url = {https://openreview.net/forum?id=uiKVKTiUYB0},
year = {2022}
}

@article{ludwig1978qualitative,
  title={Qualitative analysis of insect outbreak systems: the spruce budworm and forest},
  author={Ludwig, Donald and Jones, DD and Holling, CS},
  journal={Journal of Animal Ecology},
  volume={47},
  number={1},
  pages={315--332},
  year={1978},
  publisher={JSTOR}
}

@book{strogatz2014nonlinear,
  title={Nonlinear Dynamics and Chaos: With Applications to Physics, Biology, Chemistry, and Engineering},
  author={Strogatz, Steven H.},
  year={2014},
  publisher={CRC Press},
  edition={2nd},
  isbn={978-0813349107}
}

@article{gardner_construction_2000,
	title = {Construction of a Genetic Toggle Switch in Escherichia coli},
	volume = {403},
	doi = {10.1038/35002131},
	pages = {339--42},
	journal = {Nature},
	author = {Gardner, Timothy and Cantor, Charles and Collins, James},
	year = {2000-02-01},
}

@article{ANGELI2004185,
	author = {David Angeli and Eduardo D. Sontag},
	doi = {https://doi.org/10.1016/j.sysconle.2003.08.003},
	issn = {0167-6911},
	journal = {Systems \& Control Letters},
	number = {3},
	pages = {185-202},
	title = {Multi-stability in monotone input/output systems},
	url = {https://www.sciencedirect.com/science/article/pii/S0167691103002317},
	volume = {51},
	year = {2004}}

@article{li_multistability_2013,
	title = {Multistability and Its Robustness of a Class of Biological Systems},
	volume = {12},
	issn = {1558-2639},
	url = {https://ieeexplore.ieee.org/document/6578197},
	doi = {10.1109/TNB.2013.2271220},
	pages = {321--331},
	number = {4},
	journal = {{IEEE} Transactions on {NanoBioscience}},
	author = {Li, Yuanlong and Lin, Zongli},
	year = {2013-12},
}

@misc{hsu2026jointoptimizationapproachidentifying,
      title={A joint optimization approach to identifying sparse dynamics using least squares kernel collocation}, 
      author={Alexander W. Hsu and Ike Griss Salas and Jacob M. Stevens-Haas and J. Nathan Kutz and Aleksandr Aravkin and Bamdad Hosseini},
      year={2026},
      eprint={2511.18555},
      archivePrefix={arXiv},
      primaryClass={stat.ME},
      url={https://arxiv.org/abs/2511.18555}, 
}

@book{Goldstein2002,
  title     = {Classical Mechanics},
  author    = {Goldstein, Herbert and Poole, Charles P. and Safko, John L.},
  year      = {2002},
  edition   = {3rd},
  publisher = {Addison-Wesley},
  address   = {San Francisco, CA},
  isbn      = {978-0201657029}
}

@article{Tsamardinos2022,
  author = {Tsamardinos, Ioannis and Greasidou, Elissavet and Borboudakis, Georgios},
  title = {Don’t lose samples to estimation: specific-to-general transfer learning for small-sample classification},
  journal = {Patterns},
  volume = {3},
  number = {10},
  pages = {100612},
  year = {2022},
  doi = {10.1016/j.patter.2022.100612}
}

@article{Bradshaw2023,
  author = {Bradshaw, Tyler J. and Huemann, Zachary and Hu, Junjie and Rahmim, Arman},
  title = {A Guide to Cross-Validation for Artificial Intelligence in Medical Imaging},
  journal = {Radiology: Artificial Intelligence},
  volume = {5},
  number = {4},
  year = {2023},
  doi = {10.1148/ryai.220232}
}

@inproceedings{Kingma2014AdamMA,
    title={Adam: A Method for Stochastic Optimization},
    author={Diederik P. Kingma and Jimmy Ba},
    booktitle={3rd International Conference on Learning Representations (ICLR)},
    year={2015},
    url={http://arxiv.org/abs/1412.6980}
}

@misc{pnnl_ftnode,
  author       = {Griss Salas, Ike and King, Ethan},
  title        = {{FTNODE}: Modeling and Control of Asymptotically Stable Systems Using a Foward Tracking Neural ODE Approach.},
  year         = {2026},
  publisher    = {GitHub},
  howpublished = {\url{https://github.com/pnnl/FTNODE}},
}

\appendix
\section{Code and data availability}
The source code used to generate the results are open-source and available on GitHub at \url{https://github.com/pnnl/FTNODE}

\end{document}